\documentclass[11pt,a4paper,headsepline,headinclude]{scrartcl}
\usepackage[utf8]{inputenc}
\usepackage[english]{babel}
\usepackage{indentfirst}
\usepackage{misccorr}
\usepackage{graphicx}
\usepackage{amsmath}
\usepackage{amsthm}
\usepackage{amssymb}
\usepackage{float}
\usepackage[sort,numbers]{natbib}
\usepackage{bm}
\usepackage{comment}
\usepackage[normalem]{ulem}
\usepackage[left=3cm,right=3cm,
top=3cm,bottom=3cm,bindingoffset=0cm]{geometry}
\usepackage{pgfplots}
\usepackage{tikz}
\usepackage{enumitem}
\usepackage[unicode, pdftex]{hyperref}
\usepackage{subcaption}
\usepackage{multirow}

\theoremstyle{theorem}
\newtheorem{lemma}{Lemma}
\newtheorem{theorem}{Theorem}
\newtheorem{proposition}{Proposition}
\newtheorem{corollary}{Corollary}

\theoremstyle{definition}
\newtheorem{defin}{Definition}
\newtheorem{notation}{Notation}
\newtheorem{example}{Example}
\newtheorem{remark}{Remark}

\usepackage{authblk}
\usepackage{algorithm}
\usepackage[noend]{algpseudocode}
\newcommand{\Res} {\mathrm{Res}\,}
\newcommand{\Disc} {\mathrm{Disc}\,}

\newcommand{\sep} {\mathrm{sep}\,}
\newcommand{\init} {\mathrm{init}\,}
\newcommand{\ord} {\mathrm{ord}\,}
\newcommand{\rk} {\mathrm{rk}\,}

\def\Return{\State\textbf{return}\ }
\DeclareMathOperator{\wdeg}{wdeg}
\DeclareMathOperator{\bbeta}{\bm{\beta}}
\DeclareMathOperator{\balpha}{\bm{\alpha}}
\DeclareMathOperator{\bgamma}{\bm{\gamma}}
\DeclareMathOperator{\bomega}{\bm{\omega}}
\DeclareMathOperator{\bx}{\mathbf{x}}
\DeclareMathOperator{\bz}{\mathbf{z}}
\DeclareMathOperator{\by}{\mathbf{y}}
\DeclareMathOperator{\bX}{\mathbf{X}}
\DeclareMathOperator{\bu}{\mathbf{u}}
\DeclareMathOperator{\bg}{\mathbf{g}}
\DeclareMathOperator{\KK}{\mathbb{K}}
\DeclareMathOperator{\ba}{\mathbf{a}}

\DeclareMathOperator{\lt}{LT}

\DeclareMathOperator{\supp}{supp}

\newcommand{\Ra}[1]{\color{black} {#1}}
\newcommand{\Rb}[1]{\color{black} {#1}}

  \usepackage{coffeestains}
  \usepackage{lipsum}

\usepackage[draft=false]{scrlayer-scrpage}
\lohead{{\small \headertitle \,\, \rightmark}}
\rohead{{\small \headerauthors}}

\usepackage[textsize=tiny]{todonotes}

\newcommand{\ideal}[1]{\ensuremath{( #1 )}}
\newcommand{\tuple}[1]{\ensuremath{[#1]^{T}}}

\title{Projecting dynamical systems\\ via a support bound}

\newcommand{\headertitle}{{\normalfont%
 Projecting dynamical systems via a support bound
}}
\newcommand{\headerauthors}{%
 Y. Mukhina,
 G. Pogudin%
}

\author[1]{Yulia Mukhina}
\author[1]{Gleb Pogudin}
\affil[1]{\small LIX, CNRS, \'Ecole polytechnique, Institute Polytechnique de Paris, Paris, France\\
\texttt{\{yulia.mukhina,gleb.pogudin\}@polytechnique.edu}}
\date{}

\begin{document}

\maketitle

\begin{abstract}
    For a polynomial dynamical system, we study the problem of computing the minimal differential equation satisfied by a chosen coordinate (in other words, projecting the system on the coordinate).
    This problem can be viewed as a special case of the general elimination problem for systems of differential equations and appears in applications to modeling and control.

    We give a bound for the Newton polytope of such minimal equation.
    Our bound depends on the dimension of the model and the degrees $d$ and $D$ of the polynomials defining the dynamics of the chosen coordinate and the remaining coordinates, respectively.
    We show that our bound is sharp if $d \leqslant D$ or the model is planar.
    We further use this bound to design an algorithm for computing the minimal equation following the evaluation-interpolation paradigm.
    We demonstrate that our implementation of the algorithm can tackle problems which are out of reach for the state-of-the-art software for differential elimination.


\end{abstract}

\section{Introduction}

In this paper, we study the elimination problem for a class of differential equations.
In general, the \emph{elimination problem} is posed for a system of equations (linear, polynomial, differential, etc.)
\begin{equation}\label{eq:general_system}
  f_1(\mathbf{x}, \mathbf{y}) = \cdots = f_n(\mathbf{x}, \mathbf{y}) = 0
\end{equation}
in two groups of unknowns $\mathbf{x} = \tuple{x_1, \ldots, x_s}$ and $\mathbf{y} = \tuple{y_1, \ldots, y_\ell}$. 
The goal is to describe nontrivial equations $g(\mathbf{y}) = 0$ in $\mathbf{y}$ only,
which hold for every solution of the system.
Classical elimination methods include Gaussian elimination for linear equations and resultants and Gr\"obner bases for polynomial elimination.

The elimination problem for a system of differential equations was posed by J.~Ritt, one of the founders of differential algebra, in the 1930s~\cite{Ritt}.
In the past decades, differential analogues have been developed for the most popular approaches to polynomial elimination including differential resultants~\cite{Li2015,Rueda2013,Rueda2016,McCallum2018,CarraFerro1997}, differential Gr\"obner bases~\cite{Ollivier1991,Zobnin2006,Ferro2007}, and various versions of differential triangular sets~\cite{WANG2002,Boulier1,Boulier2,Hubert2000,Hubert2003b,diff_Thomas,Robertz2014}.
Several of these algorithms were turned into software implementations~\cite{TDDS,blad} and found applications in different domains~\cite{Boulier2007,Pascadi2021,Gerdt2019}.

Many systems naturally arising in applications to modeling and control are in the \emph{the state-space form} 
\begin{equation}\label{eq:general_state_space}
\bx' = \bg(\bx, \bu),
\end{equation}
where $\bx$ and $\bu$ are two sets of differential variables corresponding to the internal state of the system and external forces, respectively.
This motivated more recent developments~\cite{Dong2023,Meshkat2012} of differential elimination algorithms tailored to the systems of the form~\eqref{eq:general_state_space}.
Such methods could outperform the general-purpose state-of-the-art elimination software~\cite[Section~6.3]{Dong2023}.
However, they were still based on polynomial reductions and iterated resultant computations and ultimately suffered, as well as the more general differential elimination methods mentioned above, from \emph{intermediate expression swell}.

One way to address the issue of the intermediate expression swell is to use an \emph{evaluation-interpolation} approach: estimate the support of the polynomial(s) of interest and use sampled points on the corresponding variety to recover the coefficients by solving a linear system.
This paradigm has been employed successfully for several cases of polynomial elimination, for example, to perform implicitization~\cite{Marco2002,Emiris2013} or compute likelihood equations~\cite{Tang2019}. 
Two key ingredients making this approach work are a \emph{bound for the support} of the result and the ability to \emph{sample points} on the variety.
For an ODE system~\eqref{eq:general_state_space}, the latter is readily available, for example, in the form of the efficient algorithms for computing truncated power series solutions~\cite{bostan2006fast,van2010newton}.
Thus, it remains to produce a bound for the support, and this is the key question studied in the present paper.

More precisely, we consider a polynomial dynamical system, that is, an ODE system of the form
\begin{equation}\label{eq:our_case}
\bx' = \bg(\bx),
\end{equation}
where $\bx = \tuple{x_1, \ldots, x_n}$ and $\bg \in \KK[\bx]^n$ is a vector of polynomials in $\bx$ over a field $\KK$ of zero characteristic.
Compared, to the more general form~\eqref{eq:general_state_space}, we restrict ourselves to the models without external inputs and having polynomial dynamics.
The elimination problem we consider is to eliminate all the variables but one, say $x_1$.
Many elimination questions of practical importance fall in this class (see, e.g., examples from~\cite{Dong2023,Harrington2016}).
The set of all relations involving only $x_1$ and its derivatives is uniquely defined by its minimal differential equation (see Section~\ref{sec:preliminaries} for details), so we aim at computing this equation.
A related question of computing the minimal differential equation for a primitive element for a class of models including~\eqref{eq:our_case} was treated in~\cite{DAlfonso2006} from the perspective of theoretical complexity (see also Remark~\ref{rem:resolvent_comparison}).

The contribution of the present paper is threefold.
We give {\Rb the first practical bound} for the Newton polytope of the minimal differential equation satisfied by the $x_1$-coordinate of any trajectory of a polynomial dynamical system~\eqref{eq:our_case} (Theorem~\ref{theorem_general_specialized}).
The bound depends only on $d := \deg g_1$ and $D := \max_{2 \leqslant i \leqslant n} \deg g_i$.

Second, we show that our bound is sharp in ``more than half of the cases'': if $d \leqslant D$ (Theorem~\ref{thm:sharp}) or if $n = 2$ (Theorem~\ref{thm:2d}).
This contrasts with other bounds related to the differential elimination problem, e.g.~\cite{Gustavson2018,Ovchinnikov2021,Grigorev1989}.
In the cases when the bound is not sharp, we give numerical evidence that it is often quite accurate.
Finally, we use the bound to design a differential elimination algorithm following the evaluation-interpolation paradigm which can perform elimination for the cases which were out of reach for the existing software.
The algorithm does not perform any heavy polynomial manipulations, and most of the runtime is spent on solving a linear system of the dimension equal to the number of points in the predicted Newton polytope.
The proof-of-concept implementation of our algorithm in Julia is available at \url{https://github.com/ymukhina/Loveandsupport.git}.

The idea behind the proof of the bound is to reduce the differential elimination problem to a polynomial elimination problem.
Based on the analysis of the support of the produced polynomial system, we bound the support of the result of polynomial elimination 
 via several applications of the B\'ezout theorem with a specially chosen set of weights.
We would like to mention that an alternative approach could be to reduce the problem to an instance of the implicitization problem.
Then tropical methods could be used to bound the support~\cite{Sturmfels2007,rose2023tropicalimplicitizationrevisited,Emiris2012}.
However, the bounds produced this way are not sharp already in the planar case.
See Section~\ref{sec:compare_tropical} for a discussion.

The proof of the sharpness of the bound combines techniques from algebra and logic
with an interesting connection between the elimination problem and the identifiability of the system in the sense of control theory: 
in order to reach the bound we had to construct an infinite series of ODE systems with provably good identifiability properties.

The paper is organized as follows.
Section~\ref{sec:preliminaries} contains the notions and facts from differential algebra used to state the main theoretical results.
These results are stated in Section~\ref{sec:main_results}.
In Section~\ref{sec:discussion} we give numerical evidence of the accuracy of our bound (Section~\ref{sec:bound_data}), discuss the potential and limitations of our method of establishing the bound (Section~\ref{sec:mixed_fiber}), and compare the results with a potential alternative approach via tropical implicitization (Section~\ref{sec:compare_tropical}).
Sections~\ref{sec:proofs:general}, \ref{sec:proofs_bound}, and~\ref{sec:proofs_sharp} contain the proofs of the main theoretical results.
We use our bound to design an elimination algorithm in Section~\ref{sec:algorithm}, and we showcase its performance on challenging examples in Section~\ref{sec:hard_examples}.

\paragraph{Acknowledgements}
The authors are grateful to Joris van der Hoeven for numerous stimulating discussions (including a shorter proof of Lemma~\ref{lem::Puiseux}).
The authors also thank Carlos D'Andrea for proposing the idea of the proof of Proposition~\ref{prop:specific_to_generic}, Anton Leykin, Wei Li, Sonia L. Rueda, Fran\c{c}ois Boulier, Kemal Rose, and Rafael Mohr for helpful conversations.
The authors would like to thank the referees for careful and diligent reading and numerous helpful suggestions.
This work has been supported by the French ANR-22-CE48-0008 OCCAM and ANR-22-CE48-0016 NODE projects.

\section{Preliminaries}\label{sec:preliminaries}

\begin{defin}
A \emph{differential ring} $(R, \;')$ is a commutative ring with a derivation $': R \rightarrow R$, that is, a map such that, for all $a,b \in R$, $(a+b)' = a' +b'$ and $(ab)' = a'b + ab'$. A differential field is a differential ring which is a field. 
For $i > 0$, $a^{(i)}$ denotes the $i$-th order derivative of $a \in R$.
\end{defin}

Throughout the paper, $\KK$ stands for a field of zero characteristic equipped with the zero derivation.

\begin{defin}
    Let $R$ be a differential ring. An ideal $I \subset R $ is called a \emph{differential ideal} if $a' \in I$ for every $a \in I$.
\end{defin}

Let $x$ be an element of a differential ring. 
We introduce notation $x^{(\infty)} := \{x,x',x'',x^{(3)},\ldots\}.$

\begin{defin}
    Let $R$ be a differential ring. Consider a ring of
polynomials in infinitely many variables
\[
R[x^{(\infty)}] := R[x, x', x'', x^{(3)},\ldots]
\]
and extend \cite[§ 9, Prop. 4]{bourbaki1950elements} the derivation from $R$ to this ring by $(x^{(j)})':= x^{(j+1)}$. The resulting differential ring is called the \emph{ring of differential polynomials in $x$ over $R$}.
The ring of differential polynomials in several variables is defined by iterating this construction.
\end{defin}

\begin{notation}
    One can verify that $\ideal{f_1^{(\infty)},\ldots, f_s^{(\infty)}}$ is a differential ideal 
    for every $f_1, \ldots , f_s \in  R[x_1^{(\infty)},\ldots, x_n^{(\infty)}] $. Moreover, this is the minimal differential ideal containing $f_1, \ldots , f_s$,
    and we will denote it by $\ideal{f_1,\ldots, f_s}^{(\infty)}$.
\end{notation}

\begin{notation}[Saturation]
Let $I$ be an ideal in {\Rb a} ring $R$, and $a \in R$. Denote
$$
I : a^{\infty} := \{ b \in R\; | \; \exists N \colon\, a^Nb \in I\}.
$$
The set $I : a^{\infty}$  is also an ideal in $R$.
If $I$ is a differential ideal, than $I : a^{\infty}$ is also a differential ideal.
\end{notation}

 \begin{defin}
    Let $P \in \mathbb{K}[\bx^{(\infty)}]$ be a differential polynomial in $\bx = \tuple{x_1, \ldots, x_n}$.
    \begin{enumerate}
        \item For every $1 \leqslant i \leqslant n$, we will call the largest $j$ such that $x_i^{(j)}$ appears in $P$ the \emph{order} of $P$ 
        respect to $x_i$ and denote it by $\ord_{\!x_i} P$; if $P$ does not involve $x_i$, we set $\ord_{\!x_i} P := -1$.

        \item For every $1 \leqslant i \leqslant n$ such that $x_i$ appears in $P$, the \emph{initial} of $P$ with respect to $x_i$ is the
        leading coefficient of $P$ considered as a univariate polynomial in $x_i^{(\ord_{\!x_i} P)}$. 
        We denote it by $\init_{\!x_i} P$.
        
        \item The \emph{separant} of $P$ with respect to $x_i$ is
        $$
        \sep_{\!x_i} P := \dfrac{\partial P}{\partial x_{i}^{(\ord_{\!x_i} P)}} .
        $$
   
    \end{enumerate}
    \end{defin}

    The elimination problem we study in this paper is to eliminate all the variables in the system \eqref{eq:our_case} except one, say $x_1$. 
    In other words, we want to describe a differential ideal
	
	\begin{equation} \label{id}
		I = \ideal{x'_1 - g_1(\mathbf{x}), \ldots , x'_n - g_n(\mathbf{x})}^{(\infty)} \cap \mathbb{K}[x_1^{(\infty)}].
	\end{equation}

     Since the differential ideal $\ideal{x'_1 - g_1(\mathbf{x}), \ldots , x'_n - g_n(\mathbf{x}) }^{(\infty)}$ is prime~\cite{Moog1990} (see also~\cite[Lemma~3.2]{hong2020global}), the {\Rb elimination} ideal is prime as well.

    \begin{defin}
        The \emph{minimal polynomial} $f_{\min}$ of the prime ideal~\eqref{id} is a polynomial in~\eqref{id} of the minimal order and then the minimal total degree.
        It is unique up to a constant factor \cite[Proposition 1.27]{pogudin2023differential}.
    \end{defin}

    \begin{proposition}[{{\cite[Proposition~1.15]{pogudin2023differential}}}]\label{prop::idealstructure} 
        The prime ideal~\eqref{id} is uniquely determined by its minimal polynomial $f_{\min}$. 
        More precisely:
        \begin{equation*} \label{diff_ideal}
            I = \ideal{f_{\min}}^{(\infty)} : (\sep_{\!x_1}(f_{\min}) \,\init_{\!x_1}(f_{\min}))^{\infty}.
        \end{equation*}
    \end{proposition}

    \begin{example} \label{ex::boytoy}
        For a toy example of such representation consider the following model:
	  	\begin{equation*}
      			\begin{cases}
				x_1' = x_2, \\
				x_2' = - x_1.
			\end{cases}
		\end{equation*}
	$f_{\min}$ can be  obtained by $ (x_1 ' - x_2)' + (x_2' + x_1) = x_1'' + x_1$. Thus, $ f_{\min} = x_1'' + x_1$ and $I = \ideal{x_1''  + x_1}^{(\infty)}$.
    \end{example}

    The following example shows the importance of taking the saturation in Proposition~\ref{prop::idealstructure}.
    \begin{example}
    Consider
        \begin{equation*}
      			\begin{cases}
				x_1' = x_2^2, \\
				x_2' = x_1.
			\end{cases}
		\end{equation*}
    Using Algorithm~\ref{alg:fmin}, which will be presented in Section~\ref{sec:algorithm}, we obtain the minimal polynomial $f_{\min} =  (x_1'')^2 - 4 x_1^2 x_1'$ for the elimination ideal $I$. 
    In this case $\init_{x_1}f_{\min} = 1$ and $\sep_{\!x_1} f_{\min} = 2 x_1''$. 

    Using $(x_1'')^2 \equiv 4x_1^2x_1' \pmod{I}$, we can rewrite $x_1^3 (f_{\min})'$ modulo $I$ as follows:   
    \[
    x_1^3 f_{\min}' = 2 x_1^3 x_1'' x_1^{(3)} - 4 x_1^5 x_1'' - 8 x_1^4 (x_1')^2  \equiv 2 x_1^3 x_1'' x_1^{(3)} - 4 x_1^5 x_1'' - \frac{1}{2} (x_1'')^4 \pmod{I}.  
    \]
    Thus, $x_1^3 f_{\min}' = 2x_1'' (x_1^3 x_1^{(3)} - 2 x_1^5 - \frac{1}{4} (x_1'')^3)$.
    Since the ideal $I$ is prime and $x_1'' \not\in I$, we have
    \[
    x_1^3 x_1^{(3)} - 2 x_1^5 - \frac{1}{4} (x_1'')^3 \in I.
    \]
    However, this polynomial does not belong to $\ideal{f_{\min}}^{(\infty)}$ 
    (that is, without saturation at $x_1''$) because $f_{\min}$ vanishes if $x_1$ is a nonzero constant, and the above polynomial does not.
    \end{example}


\section{Bound for the support and its sharpness}\label{sec:main_results}

\begin{theorem}[Bound for the support]\label{theorem_general_specialized}
    Let $g_1, \ldots, g_n$ be polynomials in $\mathbb{K}[x_1, \ldots, x_n] = \mathbb{K}[\mathbf{x}]$ 
    such that $d := \deg g_1 > 0$ and $D := \max_{2 \leqslant i \leqslant n}\deg g_i > 0$.
    Let $I := \ideal{\mathbf{x}' - \mathbf{g}}^{(\infty)}$ and let $f_{\min} \in \mathbb{K}[x_1^{(\infty)}]$ be the minimal polynomial of $I \cap \mathbb{K}[x^{(\infty)}_1]$.
    Consider a positive integer $\nu$ such that $\ord f_{\min} \leqslant \nu$ ($\nu = n$ can be always used).

    Then  for every monomial  $x_1^{e_0} (x'_1)^{e_1} \ldots (x^{(\nu)}_1)^{e_\nu}$ in $f_{\min}$ the following inequalities hold
     
    \begin{enumerate}
        \item\label{case:d1_leq_D} If $d \leqslant D$, then
            \begin{equation} \label{eq:bound_1}
                e_0 + \sum_{k=1}^{\nu} \bigl(d + (k - 1)(D - 1) \bigr) e_k \leqslant \prod_{k=1}^{\nu} \bigl(d + (k - 1)(D - 1)\bigr);
            \end{equation}
        \item\label{case:d1_ge_D} If $d > D$, then for every $0 \leqslant \ell < \nu$, we have
           \begin{equation} \label{eq:bound_2}
                \begin{split}
                     & \sum_{k=0}^{\ell} \bigl(k(D - 1) + 1 \bigr) e_k + \sum_{i=1}^{\nu - \ell} \bigl( i (d - 1) + \ell (D - 1) + 1\bigr) e_{i + \ell} \leqslant \\
                     & \quad \leqslant \prod_{k = 1}^{\ell} \bigl(d + (k - 1) (D - 1) \bigr) \prod_{i=1}^{\nu - \ell} \bigl( i (d - 1) + \ell (D - 1) + 1\bigr).
                \end{split}
             \end{equation}
    \end{enumerate}
\end{theorem}

The corresponding polytopes in the planar case $n = 2$ are shown on Figure~\ref{fig:2d}.

\begin{notation}
 Let $V_{n, d}$ be the space of polynomials of degree at most $d$ in the variables $\mathbf{x} = \tuple{x_1, \ldots, x_n}$ over the field $\mathbb{K}$.
\end{notation}

\begin{theorem}[Generic sharpness in the planar case]\label{thm:2d}
    Let $d_1, d_2 \in \mathbb{Z}_{> 0}$.
    Then there exists a nonempty Zariski open subset $U \subset V_{2, d_1} \times V_{2, d_2}$ such that, for every pair of polynomials  $\tuple{g_1, g_2} \in U$, the Newton polytope of the minimal polynomial $f_{\min}$ of the differential ideal
    \[I \cap \mathbb{K}[x_1^{(\infty)}],\quad \text{where } I = \ideal{x_1' - g_1(x_1, x_2), x_2' - g_2(x_1, x_2)}^{(\infty)}
    \]
    is the one given by Theorem~\ref{theorem_general_specialized} with $\nu = n = 2$ (see also Figure~\ref{fig:2d}).

    In this particular case, the Newton polytope of $f_{\min}$ is given by the following inequalities (on the exponents of $x_1^{e_0} (x'_1)^{e_1}(x''_1)^{e_2}$)
     \begin{enumerate}
        \item If $d_1 \leqslant d_2$, then
        \begin{equation} \label{eq0001}
        e_0 + d_1 e_1 +(d_1 + d_2 - 1)e_2 \leqslant d_1(d_1 + d_2 - 1);
        \end{equation}
        
        \item If $d_1 > d_2$, then
        \begin{equation} \label{eq0002}
            \begin{split}
                &  e_0 + d_1 e_1 +(2d_1-1) e_2 \leqslant d_1(2d_1-1),\\
                &  e_0 + d_2 e_1  +(d_1+d_2-1)e_2 \leqslant d_1(d_1+d_2-1).
            \end{split}
        \end{equation}  
    \end{enumerate}
    \end{theorem}

    \begin{figure}[H]
        \centering
        \hspace{-1cm}\begin{subfigure}[b]{0.45\textwidth}
            \hspace{-2cm}\includegraphics[scale=0.2]{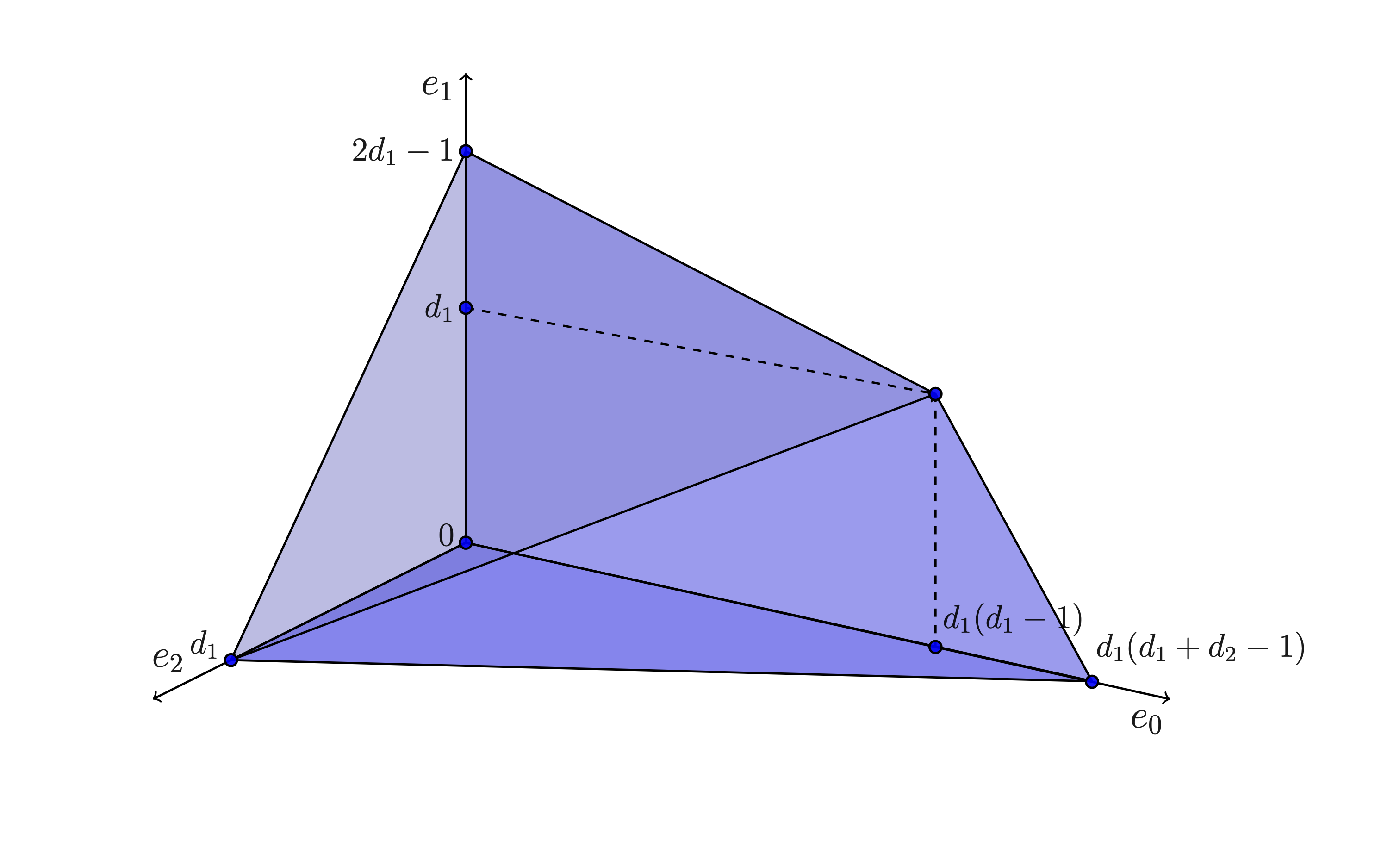}
            \caption{$d_1 > d_2$ (pyramid)} \label{ris1}
        \end{subfigure}
        \hfill
        \begin{subfigure}[b]{0.45\textwidth}
            \hspace{-1.3cm}\includegraphics[scale=0.2]{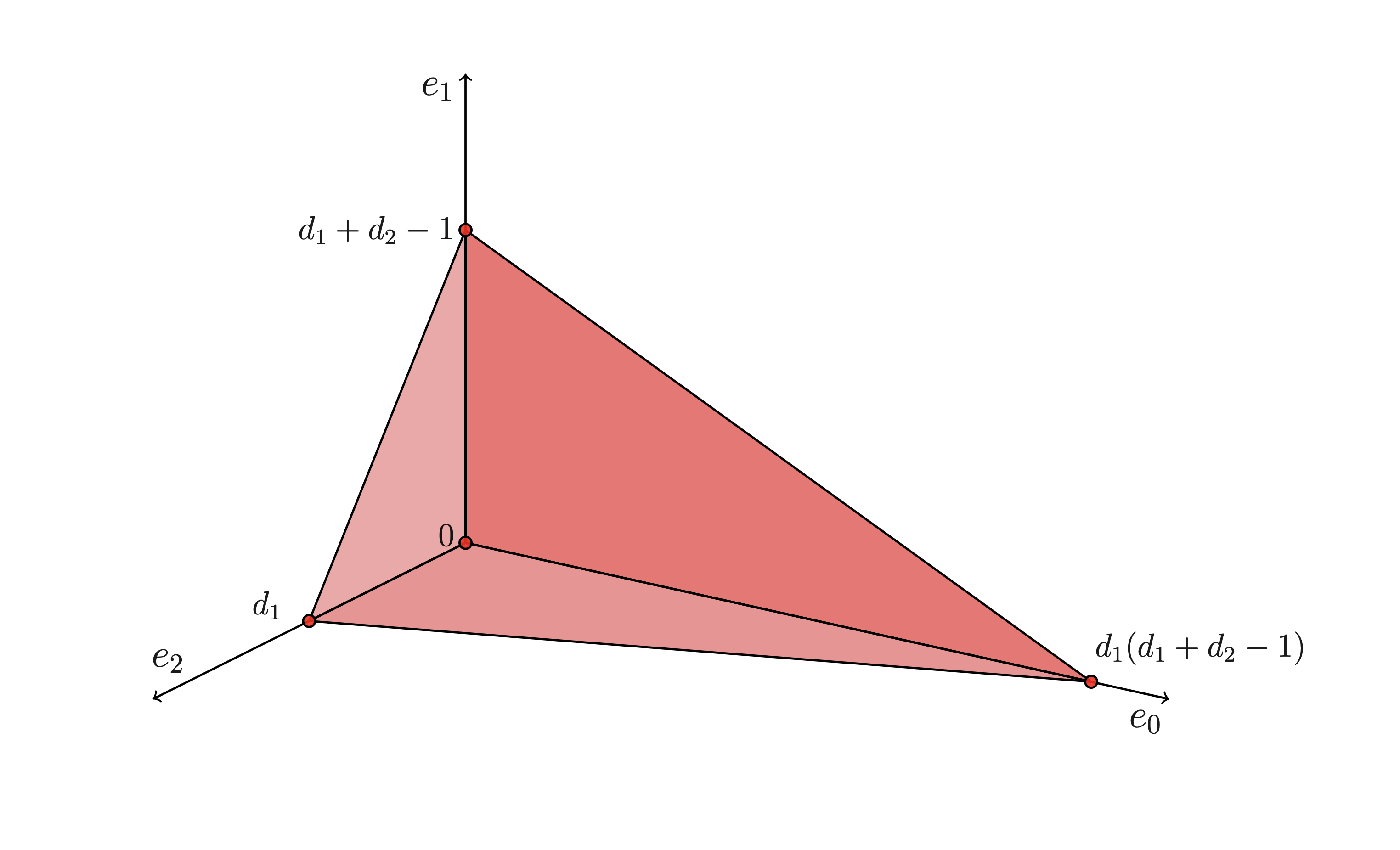}
            \caption{$d_1 \leq d_2$ (tetrahedron)} \label{ris2}
        \end{subfigure}
        \caption{Newton polytopes predicted by Theorem~\ref{theorem_general_specialized} for the planar case $n = 2$}
        \label{fig:2d}
    \end{figure}

    \begin{theorem}[Generic sharpness for $d \leqslant D$]\label{thm:sharp}
        Let $d, D, n$ be positive integers such that $d \leqslant D$.
        Then there exists a nonempty Zariski open subset $U \subset V_{n, d} \times V_{n, D}^{n - 1}$ such that, for every $\bg \in U$, the Newton polytope of the minimal polynomial of $\ideal{\bx' - \bg}^{(\infty)} \cap \KK[x_1^{(\infty)}]$ is the one given by Theorem~\ref{theorem_general_specialized} with $\nu = n$.
    \end{theorem}


\section{Discussion of the bound}\label{sec:discussion}

\subsection{Experimental evaluation of the bound's accuracy}\label{sec:bound_data}

In order to determine the accuracy of our bound, we took several triples $(n, d, D)$, and generated random dense ODE models
 $\bx' = \bg(\bx)$ of dimension $n$ with $d = \deg g_1$ and $D = \max_{2 \leqslant i \leqslant n} g_i$ by sampling coefficients uniformly at random from $[-1000, 1000] \cap \mathbb{Z}$.
Tables~\ref{tab:3d} and~\ref{tab:4d} summarize the results for $n = 3$ and $n = 4$, respectively, and contain the following columns:
\begin{itemize}
    \item \emph{\# terms in the bound}: the number of integer points inside the bound for the Newton polytope of the minimal polynomial according to the Theorem~\ref{theorem_general_specialized};
    \item \emph{\# terms in the NP of $f_{\min}$}: the number of lattice points in the Newton polytope of the actual minimal polynomial (computed by our Algorithm~\ref{alg:fmin})
    \item \emph{\# terms in $f_{\min}$}: the number of monomials in the actual minimal polynomial;
    \item \emph{\%}: the ratio between the number of monomials in $f_{\min}$ and the number of monomials in the bound from Theorem~\ref{theorem_general_specialized}. 
\end{itemize}
These numbers were consistent over several independent runs, so, with high probability, they are equal to the generic ones.
From the tables, we can observe that even if the bound is not tight, then it is still quite accurate.
The numbers also indicate that, for dense inputs, the minimal polynomial is almost dense inside its Newton polytope.
Thus, predicting the Newton polytope may lead to nearly optimal support estimates.

\begin{table}[H]
\centering
	\begin{tabular}{ |c|c|c|c|c| } 
    \hline
	\multirow{2}{*}{$[d, D]$}  & \multicolumn{3}{|c|}{\# of terms}& \multirow{2}{*}{\%} \\ \cline{2-4}
    & Theorem~\ref{theorem_general_specialized} & NP of $f_{\min}$ & $f_{\min}$ & \\
	\hline\hline	
    [2,1] & 271 & 261 & 261 & 96\%\\ 
	\hline
	[2,2] & 1292 & 1292 & 1292 & 100\%\\ 
	\hline
	[2,3] & 7875 & 7875 & 7875 & 100\%\\ 
	\hline
	[2,4] & 31757 & 31757 & 31757 & 100\% \\ 
	\hline
	[2,5] & 98771 & 98771 & 98771 & 100\%\\ 
	\hline
	\hline
	[3,1] & 9520 &  8465 &  8409  & 88\%\\ 
	\hline
	[3,2] & 25788 &  25399 & 25399  & 98\% \\ 
    \hline
    [3,3] & 65637 & 65637 & 65637 & 100\% \\
    \hline
\end{tabular}
\caption{Bound for the dimension $n = 3$}\label{tab:3d}
\end{table}
\begin{table}[H]
\centering
	\begin{tabular}{ |c|c|c|c|c| } 
	\hline
	\multirow{2}{*}{$[d, D]$}  & \multicolumn{3}{|c|}{\# of terms}& \multirow{2}{*}{\%} \\ \cline{2-4}
    & Theorem~\ref{theorem_general_specialized} & NP of $f_{\min}$ & $f_{\min}$ & \\
	\hline\hline
	[1,2] & 8189 & 8189 &  8189 & 100\% \\ 
    \hline
    [2,1] & 11021 &  10617 & 10617 & 96 \% \\ 
    \hline
\end{tabular}
\caption{Bound for the dimension $n = 4$}\label{tab:4d}
\end{table}
\begin{remark}\label{rem:resolvent_comparison}As a part of the study of differential resolvents, \cite[Theorem~36]{DAlfonso2006} establishes a degree bound for the minimal polynomial which, in our notation, can be written as $N := (\max(d, D))^{2n^2}$.
In principle, one could use this bound to {\Rb estimate} the size of the support as $\binom{N + n + 1}{n + 1}$, but this is impractical:
already for $d = D = 2$, $n = 3$ the number is $2862209$ (compared to $1292$ in Table~\ref{tab:3d}).
\end{remark}


\subsection{Potential and limitations of the present approach}\label{sec:mixed_fiber}

As described in the introduction, we obtain the bound for the result of differential elimination by constructing a polynomial elimination problem.
More precisely, the differential elimination problem for $\bx' = \bg(\bx)$ is reduced to a polynomial elimination problem for 
\begin{equation}\label{eq:polyproblem}
x_1' - g_1 = \mathcal{L}_{\bg}^\ast(x_1' - g_1) = \ldots = (\mathcal{L}_{\bg}^\ast)^{n - 1}(x_1' - g_1) = 0,
\end{equation}
where the operator $\mathcal{L}^\ast_{\bg} \colon \KK[x_1^{(\infty)}, x_2, \ldots, x_n] \to \KK[x_1^{(\infty)}, x_2, \ldots, x_n]$ will be defined in~\eqref{eq:opD}.
The proof of Theorem~\ref{theorem_general_specialized} only uses the information on the supports of~\eqref{eq:polyproblem} 
and ignores possible relations between the coefficients.
This turns out to be sufficient for establishing a sharp bound in many cases. 
For the remaining cases, one can naturally ask:
\begin{enumerate}[label = (Q\arabic*)]
    \item\label{question:refine} Is it possible to refine the bound from Theorem~\ref{theorem_general_specialized} by using only the information about the supports of the polynomial system~\eqref{eq:polyproblem}? (e.g., by analysis of its mixed fiber polytope~\cite{Esterov2008})
    \item\label{question:sharp} May such a refinement produce a sharp bound for the remaining cases? 
\end{enumerate}

In order to answer these questions, we have conducted the following experiment.
For the smallest cases, for which the bound in Theorem~\ref{theorem_general_specialized} is not sharp, namely $(n, d, D) = (3, 2, 1), (3, 3, 1)$, we take polynomials with the same Newton polytopes as in~\eqref{eq:polyproblem} but sample the coefficients randomly from $[-1000, 1000] \cap \mathbb{Z}$.
Then we perform elimination using Gr\"obner bases, compute the Newton polytope for the resulting {\Rb eliminant} $\tilde{f}_{\min}$, and compare the number of integer 
points in it { \Rb(see column ``generic eliminant for~\eqref{eq:polyproblem}'')} with 
the number we would have obtained by performing differential elimination.

The figures suggest that, the answer to~\ref{question:refine} is ``yes'', that is, there is still some room for improving the bound by only looking at the supports of~\eqref{eq:polyproblem}.
On the other hand, the answer to~\ref{question:sharp} is ``no'' meaning that the system~\eqref{eq:polyproblem} is inherently non-generic.
Thus, making the bound tight would likely require looking beyond the supports of the polynomial system~\eqref{eq:polyproblem} (recent works in this spirit include~\cite{Dickenstein2022,esterov2024engineeredcompleteintersectionsslightly}).

\begin{table}[H]
\centering
	\begin{tabular}{ |c|c|c|c| } 
 \hline
    \multirow{2}{*}{$[n, d, D]$} & \multicolumn{3}{|c|}{\# of {\Rb lattice} points in the Newton polytope}\\
	\cline{2-4}
	 & Theorem~\ref{theorem_general_specialized}  &  {\Rb generic eliminant} for~\eqref{eq:polyproblem} & ${\Rb f_{\min}}$ \\
	\hline
 \hline
	$[3, 2, 1]$ & 271 & 266 & 261 \\
    \hline
    $[3, 3, 1]$ & 9520 & 8661 & 8465 \\ 
    \hline
\end{tabular}
\caption{Comparison, in terms of the number of integer points, of the current bound (from Theorem~\ref{theorem_general_specialized}) and the true value (from $f_{\min}$) to the best possible bound one could achieve by analyzing the supports of~\eqref{eq:polyproblem}.}\label{tab:mixedfiber}
\end{table}


\subsection{On the alternative approach via tropical implicitization}\label{sec:compare_tropical}

As we have mentioned in the introduction, computing the minimal polynomial can be reduced to the implicitization problem.
We will explain this reduction on the system 
\begin{equation}\label{eq:tropical_base_example}
    \begin{cases}
        x'_1 = x_1^2 + x_1 x_2 + x_2^2 + 1,\\
        x'_2 = x_2.
    \end{cases}
\end{equation}
We can write $x_1, x_1', x_1''$ as polynomials in $x_1, x_2$ as follows:
\begin{equation}\label{eq:parametric}
\begin{split}
x_1 &= x_1, \\ 
x_1' &= x_1^2 + x_1 x_2 + x_2^2 + 1,\\
 x_1''  &=  x_1'(2 x_1 + x_2) + x_2' (2 x_2 + x_1) = (x_1^2 + x_1 x_2 + x_2^2 + 1)(2 x_1 + x_2) + x_2 (2 x_2 + x_1).
\end{split}
\end{equation}
These three equations define a two-dimensional parametric surface with local coordinates $(x_1, x_2)$ in a three-dimensional ambient space with coordinates $(x_1, x_1', x_1'')$.
The implicit equation for this surface is exactly the desired relation between $x_1, x_1', x_1''$.
Tropical implicitization~\cite{Sturmfels2007,rose2023tropicalimplicitizationrevisited} allows to produce a bound on the Newton polytope of this implicit equation, and the bound would be sharp if the parametric representation~\eqref{eq:parametric} had generic coefficients. 
However, this is not the case.
Indeed, a direct computation shows that the Newton polytope of the minimal differential equation for $x_1$ is defined by the following inequalities:
for every monomial  $x_1^{e_0} (x'_1)^{e_1} (x''_1)^{e_2}$ in $f_{\min}$ we have
\[
e_0 + e_1 + 2 e_2 \leqslant 4,\quad \text{ and } \quad e_0 + 2 e_1 + 3 e_2 \leqslant 6.
\]
On the other hand, the bound obtained by applying the tropical implicitization methods~\cite{Sturmfels2007,rose2023tropicalimplicitizationrevisited}
to~\eqref{eq:parametric} gives a larger polytope, given by 
\[
e_0 + 2 e_1 + 3 e_2 \leqslant 6.
\]
It turns out that the implicitization problem derived from an ODE using the procedure outlined above is non-generic. 
As a result, it may yield a more conservative bound on the support than the one given by Theorem~\ref{theorem_general_specialized}, 
even when the initial ODE system is randomly chosen to be dense, rather than special like the one in equation~\eqref{eq:tropical_base_example}.
We will illustrate this by taking a random dense ODE model $\bx' = \bg(\bx)$ of dimensions $n = 2$ and $n = 3$ and degrees $d = 2$ and $D = 1$ (in the notation of Theorem~\ref{theorem_general_specialized}) with the coefficients sampled uniformly at random from $[-1000, 1000] \cap \mathbb{Z}$.

In the case, $n = 2, d = 2, D = 1$, Theorem~\ref{theorem_general_specialized} yields a polytope given by 
\[
e_0 + e_1 + 2 e_2 \leqslant 4, \quad \text{ and }\quad e_0 + 2 e_1 + 3 e_2 \leqslant 6.
\]
The polytope obtained via tropical implicitization is larger since it is given only by one of inequalities above, namely, $e_0 + 2 e_1 + 3 e_2 \leqslant 6$.

Similarly, for $n = 3, d = 2, D = 1$, Theorem~\ref{theorem_general_specialized} yields
\[
e_0 + e_1 + e_2 + 2 e_3 \leqslant 8, \quad \text{ and }\quad e_0 + e_1 + 2 e_2 + 3 e_3 \leqslant 12.
\]
On the other hand, the polytope obtained via tropical implicitization is given by 
\[
e_0 + 2 e_1 + 3 e_2 + 4 e_3 \leqslant 24.
\]
In order to quantify the difference, we show the number of lattice points (that is, the number of monomial to consider) for the polytope from Theorem~\ref{theorem_general_specialized} and for the one obtained via tropical implicitization in Table~\ref{tab:tropical}

\begin{table}[htbp!]
\centering
	\begin{tabular}{ |c|c|c| } 
	\hline
    \multirow{2}{4em}{$[n, d, D]$} & \multicolumn{2}{|c|}{\# of terms} \\
    \cline{2-3}
	 & Theorem~\ref{theorem_general_specialized}  & tropical implicitization \\
	\hline
 \hline
	$[2, 2, 1]$ & 19 & 23 \\ 
    \hline
    $[3, 2, 1]$ & 271 & 1292 \\ 
    \hline
    $[3, 3, 1]$ & 9520 & 65637 \\
    \hline
\end{tabular}
\caption{Comparison with the approach via tropical elimination}\label{tab:tropical}
\end{table}


\section{Proofs: general facts and notation}\label{sec:proofs:general}

In this section we will explain how we reduce the differential elimination problem to a polynomial elimination problem.
This construction will be then used both in the proof of the bound for the support (Theorem~\ref{theorem_general_specialized}) and in the proof of its sharpness (Theorems~\ref{thm:2d} and~\ref{thm:sharp})

We will fix some notation used throughout the rest of the paper.

\begin{notation}\label{not:proofs_notation}
    Consider an ODE system $\bx' = \bg(\bx)$ with $\bx = \tuple{x_1, \ldots, x_n}$ and $g_1, \ldots, g_n \in \KK[\bx]$.
    \begin{itemize}
        \item  The differential ideal $\ideal{\bx' - \bg}^{(\infty)} \subset \KK[\bx^{(\infty)}]$ will be denoted by $I_{\bg}$.
        \item We denote the Lie derivative operator $\mathcal{L}_{\bg} \colon \KK[\bx] \to \KK[\bx]$ by $\mathcal{L}_{\bg} := \sum\limits_{i=1}^{n} g_i \frac{\partial}{\partial x_i}$.

        \item We define the operator $\mathcal{L}^\ast_{\bg} \colon \KK[x_1^{(\infty)}, x_2, \ldots, x_n] \to \KK[x_1^{(\infty)}, x_2, \ldots, x_n]$ by the formula
        \begin{equation}\label{eq:opD}
        \mathcal{L}^\ast_{\bg} := \sum\limits_{i=2}^{n} g_i \frac{\partial}{\partial x_i} +  \sum\limits_{{\Rb j=0}}^{\infty} x_1^{({\Rb j}+1)}\frac{\partial}{\partial x_1^{({\Rb j})}}.
        \end{equation}

        \item We define the reduction homomorphism $\mathcal{R}_{\bg} \colon \KK[\bx^{(\infty)}] \to \KK[\bx]$ by $\mathcal{R}_{\bg}(x_i^{(j)}) := \mathcal{L}_{\bg}^j (x_i)$.
    \end{itemize}
\end{notation}


Next in the section we consider an ODE system:
\begin{equation}\label{eq:sigma_general}
\bx' = \bg(\bx), \quad \text{ where } \bx = \tuple{x_1, \ldots, x_n} \text{ and } g_1, \ldots, g_n \in \KK[\bx].
\end{equation}

The following lemma shows that the problem of computing the minimal polynomial of the differential elimination ideal $I_{\bg} \cap \KK[x_1^{(\infty)}]$ can be reduced to a polynomial elimination problem for polynomials
\[
  x_1' - g_1, \; \mathcal{L}_{\bg}^\ast(x_1' - g_1), \ldots, \; (\mathcal{L}_{\bg}^\ast)^{\nu - 1}(x_1' - g_1),
\]
where $\nu$ is the order of the minimal differential polynomial for $x_1$.

\begin{lemma} \label{lem:ideal_is_wanted}
 For the system~\eqref{eq:sigma_general} for every $s \geqslant 0$:
    \[
    \ideal{x'_1 - g_1, \mathcal{L}^\ast_{\bg}(x_1' - g_1), \ldots, (\mathcal{L}_{\bg}^\ast)^s(x_1' - g_1) } = I_{\bg} \cap \mathbb{K}[x_1^{(\leqslant s + 1)}, x_2, \ldots, x_n].
    \]
    Furthermore, this ideal is also equal to $\ideal{ x_1' - g_1, x_1'' - \mathcal{L}_{\bg}(g_1), \ldots, x_1^{(s+1)} - \mathcal{L}^s_{\bg}(g_1)}$.
\end{lemma}

\begin{corollary}
    The minimal differential polynomial in $I_{\bg} \cap \KK[x_1^{(\infty)}]$ 
    is the generator of the principal ideal 
    $\ideal{{\Rb x_1'} - g_1, \mathcal{L}_{\bg}^\ast(x_1' - g_1), \ldots, (\mathcal{L}_{\bg}^\ast)^{\nu - 1}(x_1' - g_1)} {\Rb \cap \KK[x_1^{(\infty)}]}$.
\end{corollary}

\begin{proof}[Proof of Lemma~\ref{lem:ideal_is_wanted}]
    Denote 
    \[
    J := \ideal{x_1' - g_1, \mathcal{L}_{\bg}^\ast(x_1'  - g_1), \ldots, (\mathcal{L}_{\bg}^\ast)^s(x_1' - g_1)} \subset \KK[x_1^{(\leqslant s + 1)}, x_2, \ldots, x_n].
    \]
    We observe that, for $f \in I_{\bg} \cap \KK[x_1^{(\infty)}, x_2, \ldots, x_n]$, we have
    \[
    \mathcal{L}_{\bg}^\ast (f) = f' - \sum\limits_{i = 2}^n (x_i' - g_i) \frac{\partial f}{\partial x_i} \in I_{\bg}.
    \]
    Since $x_1' - g_1 \in I_{\bg}$, we have $J \subset I_{\bg} \cap \KK[x_1^{(\leqslant s + 1)}, x_2, \ldots, x_n]$.
    For the reverse inclusion, assume for contradiction that there exists a $p \in I_{\bg} \cap \mathbb{K}[x_1^{(\leqslant s + 1)}, x_2, \ldots, x_n]$ with $p \not\in J$. 
    We fix the monomial ordering on $\KK[x_1^{(\leqslant s + 1)}, x_2, \ldots, x_n]$ to be the lexicographic monomial ordering with 
    \[
    x_1^{(s + 1)} > x_1^{(s)} > \ldots > x'_1 > x_1 > x_2 > \ldots > x_n.
    \]
    The leading term of $(\mathcal{L}_{\bg}^\ast)^i(x_1' - g_1)$ is $x_1^{(i + 1)}$.
    Therefore, the leading terms of all generators of $J$ are distinct variables. 
    Hence this set is a Gr\"obner basis of $J$
    by the first Buchberger criterion~\cite{buchberger1979criterion}. 
    Then the result of the reduction of $p$ with respect to the Gr\"obner basis belongs to $\mathbb{K}[\mathbf{x}]$ and is distinct from zero.  
    Thus, we get a contradiction with $p \in I_{\bg}$ because $I_{\bg} \cap \KK[\mathbf{x}] = \{0\}$ by~\cite[Lemmas 3.1 and 3.2]{hong2020global}.

    Since the argument above applies verbatim if we use $\mathcal{L}_{\bg}$ instead of $\mathcal{L}_{\bg}^\ast$, $J$ coincides with $\ideal{x_1' - g_1, x''_1 - \mathcal{L}_{\bg}(g_1), \ldots, x_1^{(s+1)} - \mathcal{L}^s_{\bg}(g_1)}$.
\end{proof}

The following lemma will be used to reduce the differential ideal membership to a polynomial substitution.

\begin{lemma}[{cf. \cite[Remark 7]{DAlfonso2006}}] \label{lem::substitution}
    For the system~\eqref{eq:sigma_general} we have $\ker(\mathcal{R}_{\bg}) = I_{\bg}$.
\end{lemma}

\begin{proof}
  Let the dimension of the ODE system be $n$.
  Denote by $\prec$ the lexicographic monomial ordering on $\KK[\mathbf{x}^{\infty}]$ given by the variable ordering $x_j^{(i)}\prec x_k^{(l)}$ iff $i > l$ or $(i = l)\;\&\;(j<k)$.
  Note that the set $G:=\{ x_j^{(i)} - \mathcal{L}_{\bg}^i(g_j) \}_{1 \leqslant j \leqslant n, i \geqslant 1}$ is a Gröbner
  basis of $I_{\bg}$ w.r.t. $\prec$ by the first Buchberger criterion since the leading terms of all elements of $G$ are distinct variables {\Rb(for the notion of Gröbner bases in polynomial rings with infinite 
variables, we refer to~\cite{IimaYoshino2009}; the Buchberger 
criterion is given in~\cite[Proposition~1.13]{IimaYoshino2009}).} 
  This yields an isomorphism
  \[
  \mathbb{K}[x_1^{(\infty)},\ldots, x_n^{(\infty)}] / I_{\bg} \simeq \mathbb{K}[x_1,\ldots, x_n]
  \]
  of $\mathbb{K}$-algebras induced by sending any
  $g\in \mathbb{K}[x_1^{(\infty)},\ldots, x_n^{(\infty)}]$ to its normal form
  w.r.t. $G$.

  On the other hand, $\mathcal{R}_{\bg}(f)\in \mathbb{K}[x_1,\dots,x_n]$ and
  $f-\mathcal{R}_{\bg}(f)\in I_{\bg}$. 
  Hence $\mathcal{R}_{\bg}(f)$ must be the normal form of
  $f$ w.r.t. $G$ and $\prec$ because the monomials in $x_1,\dots,x_n$ form
  a $\mathbb{K}$-vector space basis of
  $\mathbb{K}[x_1^{(\infty)},\ldots, x_n^{(\infty)}] / I_{\bg}$. 
  But then
  $f\in I_{\bg}$ if and only if $\mathcal{R}_{\bg}(f) = 0$.
\end{proof}



\section{Proofs: the bound for the support}\label{sec:proofs_bound}

\subsection{Weighted B\'ezout bound}

{\Ra The key tool we use to prove the upper bound is the B\'ezout bound. 
In our setup, because of the presence of high-order derivatives, we found it beneficial to use the weighted version of the bound which we establish in this section.}
 
\begin{lemma} \label{lem:square_free}
    Let $g $ be a square-free polynomial in $\mathbb{K}[\mathbf{x}, y] = \mathbb{K}[x_1, \ldots, x_n, y]$. Then for every $\omega \in \mathbb{Z}_{\geqslant 0}$ there exists a homomorphism $\varphi$
    \[
    \varphi: \mathbb{K}[\mathbf{x}, y] \rightarrow \mathbb{K}[\mathbf{x}, z]
    \]
    with $\varphi(x_i) = x_i$ and $\varphi(y) = {\Ra{p}}$, 
    where $p \in \mathbb{K}[z]$, $\deg p = \omega$, such that the polynomial $\varphi(g)$ is square-free.
\end{lemma}

\begin{proof}
    Let us consider the cases $\omega = 0$ and $\omega > 0$ of the lemma separately.
    
    Assume $\omega = 0$. For every $x_i$ let us consider the discriminant $D_i := \Disc_{\!x_i} \; (g(\mathbf{x}, y))$. 
    Since $g$ is square-free, $D_i$ is a nonzero polynomial in $\mathbb{K}(\mathbf{x})[y]$ for every $i$. 
    Let us denote by $A$ the set of zeros $ \{ \alpha_1, \ldots, \alpha_m \}$ of the polynomials $D_1, \ldots, D_n$. 
    {\Ra Then any homomorphism $\varphi$ with $\varphi(y) \in \mathbb{K} \setminus A$ satisfied the requirement of the lemma.}

    Consider the case $\omega > 0$. 
    First of all note that since $g$ is square-free $g(\mathbf{x}, p(z))$ {\Ra cannot} be divisible by a square of a polynomial in $\mathbb{K}[\mathbf{x}]$. 
    Indeed, if that was the case then $g(\mathbf{x},y)$ would be divisible by the same square. 
    Then to ensure that $\varphi(g)$ is square-free we consider
    \[
    D := \Disc_{\!z} \bigl(g(\mathbf{x}, p(z))\bigr).
    \]
    Since $\frac{\partial}{\partial z} g(\mathbf{x}, p(z)) = Q_1 Q_2$ for $Q_1 := (\frac{\partial}{\partial y} g(\mathbf{x},y))_{y = p(z)}$ and $Q_2 := \frac{\partial}{\partial z} p(z)$, then 
    \[
    D = \Res_{\!z} (g(\mathbf{x}, p(z)), Q_1) \Res_{\!z}(g(\mathbf{x}, p(z)), Q_2). 
    \]
    Since $g$ is square-free,
    \[
    \Res_{\!z}\bigl(g(\mathbf{x}, p(z)), Q_1\bigr)=\Res_{\!z}\bigl(g(\mathbf{x}, y)_{y = p(z)}, (\frac{\partial}{\partial y} g(\mathbf{x},y))_{y = p(z)}\bigr) = \varphi \bigl(\Disc_{\!y}(g(\mathbf{x}, y))\bigr)
    \]
    is nonzero. 

    Let $\varphi(y)$ be $z^{\omega} + a$ for $a \in \mathbb{K}$. 
    Since $\mathbb{K}$ is infinite, we can choose values $a_i \in \mathbb{K}$ such that $\tilde{g}(y) := g(a_1, \ldots, a_n, y) \neq 0$. 
    Then over $\overline{\mathbb{K}}$ we have $\tilde{g}(z^\omega) =\prod\limits_{i=1}^N(z^\omega - \alpha_i)$. Let us choose $a \notin \{ \alpha_1, \ldots, \alpha_N \}$. Then $z$ does not divide $\tilde{g}(z^\omega - a)$ and hence also does not divide $g(\mathbf{x},\varphi(y))$. We conclude that
    \[
    \Res_{\!z}\bigl(g(\mathbf{x}, p(z)), Q_2\bigr) = \Res_{\!z}\bigl(g(\mathbf{x}, p(z)), \omega z^{\omega - 1}\bigr)
    \]
    is nonzero, finishing the proof.
\end{proof}

\begin{lemma} \label{lem:multivariate_square_free}
    Let $g $ be a square-free polynomial in $\mathbb{K}[\mathbf{x}, \mathbf{y}] = \mathbb{K}[x_1, \ldots, x_n, y_1, \ldots, y_k]$. Then for every $\tuple{\omega_1, \ldots, \omega_k} \in \mathbb{Z}_{\geqslant 0}^k$ there exists a homomorphism $\varphi$
    \[
    \varphi: \mathbb{K}[\mathbf{x}, \mathbf{y}] \rightarrow \mathbb{K}[\mathbf{x}, z_1, \ldots, z_k]
    \]
    with $\varphi(x_i) = x_i$ and $\varphi(y_i) = p_i(z_i)$, where $p_i \in \mathbb{K}[z_i]$ and $\deg p_i(z_i) = \omega_i$ such that the polynomial $\varphi(g)$ is square-free.
\end{lemma}
\begin{proof}
We will show this by induction on $k$. 
The base case $k=1$ follows from Lemma~\ref{lem:square_free}. 

We fix $k > 1$ and substitute $y_1,\ldots,y_{k-1}$ with $p_1(z_1), \ldots, p_{k-1}(z_{k-1})$ in $g$, by the induction hypothesis the resulting polynomial $\tilde{g}$ is square-free. 
Now the statement follows for $k$ by applying Lemma \ref{lem:square_free} to $\tilde{g}$ with $\bx$ being $\mathbf{x}\cup \{z_1,\ldots,z_{k-1}\}$ and $y$ being $y_k$. 
\end{proof}

\begin{lemma} \label{lem:total_degree}
    Let $p_1, \ldots, p_n$ be polynomials of degrees $d_1, \ldots, d_n$ in $\mathbb{K}[\mathbf{x}, \mathbf{y}]$ 
    with  $\bx = \tuple{x_1, \ldots, x_{m}}$ and $\by = \tuple{y_1, \ldots, y_k}$.
    Suppose that the ideal $I = \ideal{p_1, \ldots, p_n} \cap \mathbb{K}[\mathbf{y}]$ is principal, that is,
    $I = \ideal{g}$, and let $g = g(\mathbf{y})$ be a nonzero square-free polynomial. 
    Then 
    $\deg g \leqslant \prod_{i = 1}^n d_i$.
\end{lemma}

\begin{proof}

Let us consider the algebraic variety $X := \mathbb{V} (p_1, \ldots, p_n) \subset \mathbb{A}^{m+k}$. 
Since $X = \mathbb{V}(p_1) \cap \ldots \cap \mathbb{V}(p_n)$, then for the degree of variety $X$ 
by {\Rb~\cite[Theorem~1]{heintz1983definability}} we have 
\[
\deg X \leqslant \prod\limits_{i=1}^n \deg \mathbb{V}(p_i) = \prod\limits_{i=1}^n \deg p_i = \prod\limits_{i=1}^n d_i.
\]
Denote by $\pi: \mathbb{A}^{m+k} \rightarrow \mathbb{A}^{k}$ the projection onto the 
last $k$ components. 
Then by~\cite[Chapter 4, \S4, Theorem~3]{cox1997ideals} $\mathbb{V}(I) \subset \mathbb{A}^k$ is 
the Zariski closure of $\pi (X)$ and since $I$ is the principal ideal generated by a 
square-free polynomial $g$, we obtain $\mathbb{V}(g) =  \overline{ \pi (X)}$ and 
$\deg \overline{ \pi (X)} = \deg g$ {\Rb~\cite[Remark 2(3)]{heintz1983definability}}.

Applying~\cite[Lemma~2]{heintz1983definability} to the projection $\pi$ and the algebraic variety $X$ we obtain $\deg \overline{\pi (X)} \leqslant \deg X$. So
\[
   \deg g = \deg \overline{\pi (X)} \leqslant \deg X \leqslant \prod\limits_{i=1}^{n} d_i.\qedhere
\]  
\end{proof}


\subsection{Bound for the support in a general form}

We will derive Theorem~\ref{theorem_general_specialized} from the following more general bound.

\begin{theorem} \label{theorem_general}
    Let $g_1, \ldots, g_n$ be polynomials in $\mathbb{K}[x_1, \ldots, x_n] = \mathbb{K}[\mathbf{x}]$  with $d := \deg g_1 > 0$ and $D := \max_{2 \leqslant i \leqslant n}\deg g_i > 0$.
    Let $f_{\min} \in \mathbb{K}[x_1^{(\infty)}]$ be the minimal polynomial of $I_{\bg} \cap \mathbb{K}[x^{(\infty)}_1]$ and consider a positive integer $\nu$ such that
    $\ord f_{\min} \leqslant \nu$.
    
    Then for every vector $(\omega_1, \ldots, \omega_\nu) \in (\mathbb{Z}_{\geqslant 0})^{\nu}$ and every monomial  $x_1^{e_0} (x'_1)^{e_1} \ldots (x^{(\nu)}_1)^{e_\nu}$ in $f_{\min}$, the following inequality holds
    \begin{equation}\label{eq:theorem_general_bound}
         e_0 + \sum_{i=1}^{\nu} \omega_{i} e_i \leqslant \prod_{k=1}^\nu \max \biggl( \omega_k, d + (k-1)(D - 1), \max_{ 1 \leqslant j < k} \bigl(d + \dfrac{k-1}{j} (\omega_j - 1) \bigr) \biggr)
    \end{equation}
\end{theorem}

We will first introduce the necessary notations and prove a few preliminary lemmas that are essential for the proof of Theorem~\ref{theorem_general}.

{\Ra Further in the section} we consider the set of polynomials defined by 
    \begin{equation} \label{derivatives}
        h_1 := g_1, \quad  h_i := \mathcal{L}_{\bg}^\ast h_{i-1}.
    \end{equation}

\begin{notation} 
    For vectors $\balpha = \tuple{\alpha_1, \ldots, \alpha_n} \in \mathbb{Z}_{\geqslant 0}^n$ and $\bbeta = \tuple{\beta_1, \ldots, \beta_n} \in \mathbb{Z}^n$, we denote
    \begin{align*}
    m(\balpha, \bbeta) &:= x_1^{\beta_1} \ldots x_n^{\beta_n} (x_1')^{\alpha_1} \ldots (x_1^{(n)})^{\alpha_n},\\
    \ell_1(\balpha, \bbeta) &:= \sum\limits_{i = 1}^n \beta_i + \sum\limits_{i = 1}^n (iD - i + 1) \alpha_i,\\
    \ell_2(\balpha) &:= \sum\limits_{i = 1}^n i\alpha_i.
    \end{align*}
\end{notation}

\begin{lemma} \label{lem:base_operator}
    For every monomial $m(\balpha, \bbeta)$ {\Rb such that 
    $\ord_{x_1}m(\balpha, \bbeta) \leqslant n-1$ } and every monomial $m(\tilde{\balpha}, \tilde{\bbeta})$ in $\mathcal{L}_{\bg}^\ast(m(\balpha, \bbeta))$, the following inequalities hold:
    \[
    \ell_1(\tilde{\balpha}, \tilde{\bbeta}) \leqslant \ell_1(\balpha, \bbeta) + D - 1 \quad \text{ and } \quad \ell_2(\tilde{\balpha}) \leqslant \ell_2(\balpha) + 1.
    \]
\end{lemma}
\begin{proof}
  Note that
  \[
  \mathcal{L}_{\bg}^\ast(m(\balpha, \bbeta)) = \sum\limits_{i=2}^{n}g_i \frac{\partial}{\partial x_i} m(\balpha, \bbeta) + \sum\limits_{i=0}^{\infty} x_1^{(i+1)} \frac{\partial}{\partial x_1^{(i)}} m(\balpha, \bbeta).
  \]
  Thus, for a monomial $m(\tilde{\balpha}, \tilde{\bbeta})$ in $\mathcal{L}_{\bg}^\ast(m(\balpha, \bbeta))$ we have two options (below $\sim$ stands for the proportionality up to a constant):
  \begin{enumerate}
      \item $m(\tilde{\balpha}, \tilde{\bbeta}) \sim m_1 \frac{\partial}{\partial x_i} m(\balpha, \bbeta)$ for some $2 \leqslant i \leqslant n$ and some monomial $m_1$ in $g_i$.
      \item $m(\tilde{\balpha}, \tilde{\bbeta}) \sim x_1^{(i+1)} \frac{\partial}{\partial x_1^{(i)}} m(\balpha, \bbeta)$ for some $0 \leqslant i \leqslant n$.
  \end{enumerate}

  Consider the first option. In this case, $ \tilde{\balpha} = \balpha$ and $|\tilde{\bbeta}| \leqslant |\bbeta| + D - 1$. Thus, we get 
  \[
  \ell_1(\tilde{\balpha}, \tilde{\bbeta}) \leqslant \ell_1(\balpha, \bbeta) + D - 1 \quad \text{ and } \quad \ell_2(\tilde{\balpha}) = \ell_2(\balpha) \leqslant \ell_2(\balpha) + 1.
  \]

  Consider the second option. 
  If $i = 0$, then $m(\tilde{\balpha}, \tilde{\bbeta}) = x_1' \frac{\partial}{\partial x_1} m(\balpha, \bbeta)$. 
  So, $\tilde{\beta}_1 = \beta_1 - 1$ and $\tilde{\alpha}_1 = \alpha_1 + 1$. Hence,
  \[
  \ell_1(\tilde{\balpha}, \tilde{\bbeta}) \leqslant \ell_1(\balpha, \bbeta) + D - 1 \quad \text{ and } \quad \ell_2(\tilde{\balpha}) = \ell_2(\balpha) + 1.
  \]
  
  For $1 \leqslant i \leqslant n$ we have $m(\tilde{\balpha}, \tilde{\bbeta}) = x_1^{(i+1)} \frac{\partial}{\partial x_1^{(i)}} m(\balpha, \bbeta)$. 
  So, $\tilde{\alpha}_i = \alpha_i - 1$, $\tilde{\alpha}_{i+1} = \alpha_{i+1} + 1$ and $\tilde{\bbeta} = \bbeta$. 
  Thus,
  \[
    \ell_1(\tilde{\balpha}, \tilde{\bbeta}) \leqslant \ell_1(\balpha, \bbeta) + D - 1 \quad \text{ and } \quad \ell_2(\tilde{\balpha}) = \ell_2(\balpha) + 1.
  \]
  This concludes the proof.
\end{proof}

\begin{corollary} \label{cor:polyhedron}
    For {\Ra every $k \leqslant n$ and }every monomial $m(\balpha, \bbeta)$ in $h_k$ {\Ra (see~\eqref{derivatives})}, the following inequalities hold:
    \[
    \ell_1(\balpha, \bbeta) \leqslant d + (k - 1)(D - 1)\quad \text{ and } \quad \ell_2(\balpha) \leqslant k - 1.
    \]
\end{corollary}
\begin{proof}
    We observe that $g_1$ is a polynomial of degree at most $d$ in $x_1, \ldots, x_n$, so for every monomial $m(\balpha, \bbeta)$ in $h_1$ we have 
    \[
    \ell_1(\balpha, \bbeta) \leqslant d \quad \text{ and } \quad \ell_2(\balpha) = 0.
    \]
    
    Consider $h_k = \mathcal{L}_{\bg}^\ast h_{k-1} = (\mathcal{L}_{\bg}^\ast)^{k-1} (h_1)$.
    {\Ra Then $\ord_{x_1}h_k \leqslant k$.}
    {\Ra 
    By applying Lemma~\ref{lem:base_operator} $k-1$ times to each monomial $m(\balpha, \bbeta)$ in $h_1$, 
    we obtain the desired inequalities for every monomial $m(\tilde{\balpha}, \tilde{\bbeta})$ in $h_k$.
    }
\end{proof}

\begin{lemma} \label{lem:max}
   Consider $\bomega = \tuple{\omega_1, \ldots, \omega_n} \in \mathbb{Z}_{>0}^n$.
   Let $L\colon \mathbb{R}^n\times \mathbb{R}^n \rightarrow \mathbb{R}$ be the linear function defined by
   \[
     L(\balpha, \bbeta) := \sum_{i=1}^{n} \omega_i \alpha_i + \sum_{i=1}^{n} \beta_i.
   \]
  And let $P_k\subset \mathbb{R}^n\times \mathbb{R}^n$, $k \geqslant 1$ be the polyhedron defined by the inequalities
    \begin{equation} \label{eq:polyhedron_P_k}
         \ell_1(\balpha, \bbeta) \leqslant  d + (k-1)(D - 1), \quad \ell_2(\balpha) 
         \leqslant k - 1, \quad \alpha_i \geqslant 0, 
         {\Rb \quad \text{ and } \alpha_j = 0 \text{ for } j \geqslant k}.
    \end{equation}
    Then we have
     \[
    \max_{(\balpha, \bbeta) \in P_k} L(\balpha, \bbeta) = 
    \max \biggl(d + (k-1)(D - 1), \;\;\max_{ 1 \leqslant j < k} \bigl(d + \dfrac{k-1}{j} 
    (\omega_j - 1) \bigr) \biggr).
    \]
\end{lemma}  
\begin{proof}
{\Rb We introduce the notation $B := \sum_{i=1}^n \beta_i$ and note that $\beta$'s in $L$ and $\ell_1$ can be replaced with $B$:
\[
 L(\balpha, \bbeta) = \sum_{i=1}^{n} \omega_i \alpha_i + B, \quad \ell_1(\balpha, \bbeta) = B + \sum\limits_{i = 1}^n (iD - i + 1) \alpha_i.
\]
Since $L(\balpha, \bbeta)$ is linear and increasing in $B$, and the only inequality defining $P_k$ involving $B$ is $\ell_1(\balpha, \bbeta) \leqslant d + (k-1)(D-1)$, 
the maximum is achieved when $\ell_1(\balpha, \bbeta) = d + (k-1)(D-1)$.
We use this equality to find the optimal value $\widetilde{B}$ of $B$:
$$ \tilde{B} = d + (k-1)(D-1) - \sum_{i=1}^{n}(iD - i + 1) \alpha_i.$$
Substituting into $L(\balpha, \bbeta)$ gives us: 
\[
\tilde{L}(\balpha) = \sum_{i=1}^{n} \bigl(\omega_i - (iD - i + 1) \bigr)\alpha_i + d + (k-1)(D-1).
\]
Thus the problem reduces to maximizing $\tilde{L}(\balpha)$ 
over the polyhedron $\tilde{P}_k$ defined by:
\[ 
\sum_{i=1}^n i \alpha_i \leqslant k - 1, \quad \alpha_i \geqslant 0, \text{ and }
\alpha_j = 0 \text{ for } j \geqslant k. 
\]
Since $\tilde{L}(\balpha)$ is linear, its maximum over $\tilde{P}_k$
is attained at a vertex. The vertices of $\tilde{P}_k$ are
are exactly the origin and the $k-1$ points $v_i$, $1 \leqslant i < k$,
defined by $\alpha_i = \frac{k-1}{i}$ for some $1 \leqslant i < k$ and $\alpha_j = 0$ for $j \neq i$.
We evaluate $\tilde{L}(\balpha)$ at each vertex:
\begin{itemize}
    \item At the origin ($\alpha_i = 0$ for all $i$):
    \[ \tilde{L}(\balpha) = d + (k-1)(D-1).\]
    \item At $v_i$ for $1 \leqslant i < k$:
    \[ \tilde{L}(\balpha) = d + \frac{k-1}{i}(\omega_i - 1).\]
\end{itemize} 

Taking the maximum over all vertices yields
\[
\max_{(\balpha, \bbeta) \in P_k} L(\balpha, \bbeta) = \max \biggl(d + (k-1)(D - 1), \max_{ 1 \leqslant j < k} \bigl(d + \dfrac{k-1}{j} (\omega_j - 1) \bigr) \biggr).
\]
This concludes the proof of our claim.
}

\end{proof}

\begin{proof}[Proof of Theorem~\ref{theorem_general}]
Let us denote by $\mu$ the order of the minimal polynomial $f_{\min}$. We first recall that $\mu \leqslant n$ by~\cite[Theorem~3.16 and Corollary~3.21]{hong2020global}. 

By Lemma~\ref{lem:ideal_is_wanted}, we get
\[
\ideal{x'_1 - h_1, \ldots, x_1^{(\mu)} - h_{\mu}} \cap \mathbb{K}[x_1^{(\leqslant \mu)}] = I_{\bg} \cap \mathbb{K}[x_1^{(\leqslant\mu)}].
\]
Since $I_{\bg} \cap \mathbb{K}[x_1^{(\leqslant \mu)}] = \ideal{f_{\min}}$, the ideal $\ideal{x'_1 - h_1, \ldots, x_1^{(\mu)} - h_{\mu}} \cap \mathbb{K}[x_1^{(\leqslant \mu)}] = \ideal{f_{\min}}$ is prime and principal. 

Let $1 \leqslant k \leqslant \mu$.
For monomials $m(\balpha, \bbeta) \in h_k$ Corollary~\ref{cor:polyhedron} implies the following inequalities
\begin{equation*} 
            \begin{split}
                 \ell_1(\balpha, \bbeta) & \leqslant d + (k-1)(D - 1), \\
                 \ell_2(\balpha) & \leqslant k - 1,\\
                 \alpha_i & \geqslant 0.
            \end{split}
            \end{equation*}
We denote by $P_k$ the polyhedron defined by these inequalities.

Let us define a homomorphism 
\[
  \varphi \colon \mathbb{K}[\mathbf{x}, x_1', \ldots, x_1^{(\mu)}] \to \mathbb{K}[\mathbf{x}, z_1, \ldots, z_\mu],
\]
such that
\begin{align*}
  x_i & \mapsto  x_i, \\
  x_1^{(i)} & \mapsto p_i(z_i), \; \deg p_i(z_i) = \omega_i.
\end{align*}
Following Lemma \ref{lem:multivariate_square_free} we will chose 
$p_i$ such that $\tilde{f}_{\min}:=\varphi(f_{\min})$ is square-free.
We define $\tilde{f}_k : = \varphi(x_1^{(k)} - h_k)$ for $1 \leqslant k \leqslant \mu$.
Note that $\deg \tilde{f}_k \leqslant \max (\omega_k,  \max\limits_{(\balpha, \bbeta) \in P_k} L(\balpha, \bbeta))$, where $L(\balpha, \bbeta) := \sum_{i=1}^{n} \omega_i \alpha_i + \sum_{i=1}^{n} \beta_i$.
Then by Lemma~\ref{lem:max} we obtain
\[
    \deg \tilde{f}_k \leqslant \max \Biggl(\omega_k,  \max \bigl(d + (k-1)(D - 1), \max_{ 1 \leqslant j < k} (d + \dfrac{k-1}{j} (\omega_j - 1)) \bigr) \Biggr).
\]
Applying Lemma~\ref{lem:total_degree} with $p_i = \tilde{f}_i$ to the principal ideal $\ideal{\tilde{f}_1, \ldots, \tilde{f}_\mu} \cap \mathbb{K}[x_1, z_1, \ldots, z_\mu]$ we obtain $\deg \tilde{f}_{\min}  \leqslant \prod\limits_{i=1}^\mu \deg \tilde{f}_i$.

Applying $\varphi$ to a monomial $m = x_1^{e_0} (x'_1)^{e_1} \ldots (x_1^{(\mu)})^{e_\mu}$ in the support of $f_{\min}$ we obtain
$\varphi(m) = c x_1^{e_0} (z_1^{\omega_1})^{e_1} \ldots (z_\mu^{\omega_\mu})^{e_\mu} + q$ where $c \in \mathbb{K}^\ast$ and $\deg q \leqslant e_0 + \sum_{i=1}^{\mu} \omega_i e_i$.
Together with the obtained degree bound for $\tilde{f}_{\min}$, this gives
\begin{equation}\label{eq:theorem_general_bound_internal}
e_0 + \sum_{i=1}^{\mu} \omega_i e_i \leqslant \prod_{k=1}^\mu \max \biggl( \omega_k, d + (k-1)(D - 1), \max_{ 1 \leqslant j < k} \bigl(d + \dfrac{k-1}{j} (\omega_j - 1) \bigr) \biggr).
\end{equation}
Furthermore, for every monomial  $x_1^{e_0} (x'_1)^{e_1} \ldots (x_1^{(\nu)})^{e_\nu}$ in $f_{\min}$, we have $e_{\mu + 1} = \ldots = e_{\nu} = 0$, so the left-hand side of~\eqref{eq:theorem_general_bound} is equal to the left-hand side of~\eqref{eq:theorem_general_bound_internal}.
The right-hand side of~\eqref{eq:theorem_general_bound} differs from the right-hand side of~\eqref{eq:theorem_general_bound_internal} by extra factors $\geqslant 1$.
Thus, \eqref{eq:theorem_general_bound_internal} implies~\eqref{eq:theorem_general_bound}.
\end{proof}


\subsection{Proof of Theorem~\ref{theorem_general_specialized}}

\begin{proof}[Proof of Theorem~\ref{theorem_general_specialized}]
    We will deduce the result from Theorem~\ref{theorem_general} through an appropriate choice of the vectors $\bomega \in \mathbb{Z}_{\geqslant 0}^{\nu}$.
    To this end, for $\bomega = \tuple{\omega_1, \ldots, \omega_\nu} \in \mathbb{Z}_{\geqslant 0}^{\nu}$ and $1 \leqslant k \leqslant \nu $, we denote (cf.~\eqref{eq:theorem_general_bound}) 
    \begin{align*}
         m_k(\bomega) & := \max_{ 1 \leqslant j < k} \bigl(d + \frac{k-1}{j} (\omega_j - 1)\bigl),\\
        M_k(\bomega) & := \max \bigl(\omega_k, d + (k-1)(D-1), m_k(\bomega)\bigr).
    \end{align*}       
    
    \emph{Assume $d \leqslant D$.}
    Take $\bomega$ with
    \[
    \omega_i := d + (i - 1)(D - 1) \text{ for } i = 1, \ldots, \nu.
    \] 
    Let $1 \leqslant k \leqslant \nu$.
    Then, for each $1 \leqslant i < k \leqslant \nu$, we have
    \[
    d + \frac{k-1}{i}(\omega_i - 1) = d + \frac{k-1}{i}(d + (i - 1)(D - 1) - 1) \leqslant d + (k-1)(D - 1),
    \]
    so $m_k(\bomega) \leqslant d + (k-1)(D - 1)$ and $M_k(\bomega) = d + (k-1)(D-1)$. Applying Theorem~\ref{theorem_general} to $\bomega$ we obtain the inequality \eqref{eq:bound_1}.

    \emph{Assume $d > D$.}
    Fix $0 \leqslant \ell < \nu$ and take $\bomega$ such that  
    \begin{align*}
              \omega_i & := i ( D - 1 ) + 1 \text{ for } i = 1,\ldots,\ell, \\ 
              \omega_i & := (i - \ell)(d - 1) + \ell (D - 1) + 1 \text{ for } i = \ell + 1, \ldots, \nu.
    \end{align*}  
    Let $1 \leqslant k \leqslant \ell$. Then
    \[
     d + \dfrac{k-1}{i}(\omega_i - 1) = d + (k - 1)(D - 1) = m_k(\bomega),
    \]
    and  
    \[
    d + (k-1)(D-1) > (D-1) + 1 + (k-1)(D-1) = k(D-1) + 1 = \omega_k.
    \]
    Thus, we obtain
    \[
    M_k(\bomega) = d + (k-1)(D-1) \text{ for } k = 1, \ldots, \ell.
    \]

    Assume now that $\ell + 1 \leqslant k \leqslant \nu$. 
    Then we have
    \[
    d + \frac{k-1}{j}(\omega_j - 1) = d + \frac{k-1}{j} \bigl( (j-\ell)(d - 1) + \ell (D -1) \bigr) > d + (k-1)(D - 1).
    \]
    For $\ell+1 \leqslant i \leqslant j < k$ we have
    \[
    \frac{\omega_i - 1}{i} = \frac{i (d - 1) - \ell (d - D)}{i} \leqslant \frac{j (d - 1) - \ell (d - D)}{j} = \frac{\omega_j - 1}{j}.
    \]
    Hence $m_k(\bomega)= d + \omega_{k-1} - 1 = \omega_k$ and
    \[
    M_k(\bomega) = (k - \ell)(d - 1) + \ell (D-1) + 1.
    \]
    We apply Theorem~\ref{theorem_general} to the constructed vector $\bomega$, and obtain the $\ell$-th inequality in~\eqref{eq:bound_2}.
\end{proof}


\section{Proofs: sharpness of the bound}\label{sec:proofs_sharp}

The proof of the sharpness is organized as follows. 
{\Ra First, in Section~\ref{sec:generic}, we prove Proposition~\ref{prop:specific_to_generic}, 
which can be used to prove generic sharpness if (not necessarily generic) sharpness has been established, i.e., it suffices to find a single specific ODE system attaining the bound.}
Then we prove Theorem~\ref{thm:sharp} by constructing a family of ODE models satisfying $d \leqslant D$ (namely, shifts of~\eqref{eq:sigma}) on which the bound is reached. This is achieved by counting the roots of certain polynomial systems obtained as truncations of the original differential ideal and showing that these roots do not collide under the projection on $x_1$ and its derivatives.
This gives us the desired lower bounds on the degrees.
The root counting part of the argument relies on properties of some {\Ra Vandermonde} type systems which we establish in Section~\ref{sec:vandermonde}.
The projection part of the argument follows from the fact that the considered ODE models are globally observable which is established by complex-analytic examination of the solutions (thanks to the choice of ODE model, {\Ra their} solutions admit a closed form representation) (see Section~\ref{sec:proof_general_sharp}).

The proof of Theorem~\ref{thm:2d} (see Section~\ref{sec:planar_sharpness}) also exhibits specific ODE models reaching the bound.
The presence of the desired monomials is proved using the resultant representation of the minimal polynomial.


\subsection{From sharpness to generic sharpness}
\label{sec:generic}

{\Ra In this section, we prove a proposition that allows us to relate the existence of a specific ODE system attaining 
the bound for the support of its minimal polynomial to generic sharpness, as stated in Theorems~\ref{thm:2d} and~\ref{thm:sharp}.}
The idea of the proof of this proposition was suggested to us by Carlos D'Andrea.

We fix the ground field $\KK$.
For a polytope $\mathcal{P} \subset \mathbb{R}_{\geqslant 0}^n$, by $V(\mathcal{P})$ we will denote a vector space of polynomials $g \in \KK[\bx]$ (where $\bx = \tuple{x_1, \ldots, x_n}$) with support included $\mathcal{P}$. 

\begin{proposition}\label{prop:specific_to_generic}
    Let $\mathcal{P}_1, \ldots, \mathcal{P}_n$ be polytopes in $\mathbb{R}_{\geqslant 0}^n$.
    Let $\bx' = \bg(\bx)$ be an ODE system such that $g_i \in V(\mathcal{P}_i)$ 
    for every $1 \leqslant i \leqslant n$ {\Ra and} such that {\Ra the minimal polynomial of the 
    corresponding elimination ideal $\bigl(\bx'- \bg(\bx) \bigr)^{(\infty)} \cap \KK[x_1^{(\infty)}]$} is of order $n$.
    We denote the Newton polytope of $f_{\min}$ by $\mathcal{N}$.

    Then there exists a nonempty Zariski open $U \subset V(\mathcal{P}_1) \times \cdots \times V(\mathcal{P}_n)$ such that, 
    for every $\bg^\ast \in U$, the order of the minimal polynomial for $x_1$ in the system $\bx' = \bg^\ast(\bx)$ is $n$ 
    and the Newton polytope of this minimal polynomial contains a nonnegative shift of $\mathcal{N}$.
\end{proposition}

\begin{proof}
    We set $V := V(\mathcal{P}_1) \times \cdots \times V(\mathcal{P}_n)$.
    Consider $\bg^\ast \in V$.
    Then, by Lemma~\ref{lem::substitution} the minimal polynomial for $x_1$ has order $n$ if and only if $x_1, \mathcal{L}_{\bg^\ast}(x_1), \ldots,  \mathcal{L}_{\bg^\ast}^{n - 1}({\Rb x_1})$ are algebraically independent.
    This is equivalent to the fact that the Jacobian of $x_1, g_1^\ast, \ldots, \mathcal{L}_{\bg^\ast}^{n - 1}({\Rb x_1})$ is nonsingular by~\cite[Proposition 2.4]{ehrenborg1993apolarity}.
    This nonsingularity condition defines an open subset $U_1 \subset V$.
    Since $\bg \in U_1$, this subset is nonempty.

    Let $N = \dim V$ and fix a basis for $V$. For the rest of the proof, we will identify $N$-dimensional vectors over any field $\mathbb{L} \supset \KK$ with tuples $\mathbf{q} \in \mathbb{L}[\bx]^n$ with $\operatorname{supp}(q_i) \subset \mathcal{P}_i$ for every $ i = 1, \ldots, n$.
    We introduce new variables $a_1, \ldots, a_N, \varepsilon$ and denote $\ba := \tuple{a_1, \ldots, a_N}$. 
    We will write $\ba(\bx)$ for the corresponding element of $\KK(\ba)[\bx]^n$.
    We consider an ODE system $\bx' = \bg + \varepsilon\ba(\bx)$ and denote the minimal polynomial for $x_1$ in this system by $\tilde{f}_{\min}$.
    By clearing denominators, if necessary, we will assume that $\tilde{f}_{\min}\in \KK[\ba, \varepsilon, x_1^{(\infty)}]$.
    We consider the expansion of $\tilde{f}_{\min}$ with respect to $\varepsilon$:
    \[
    \tilde{f}_{\min} = f_0 \varepsilon^s + \mathcal{O}(\varepsilon^{s + 1}).
    \]
    By Lemma~\ref{lem::substitution}, $\tilde{f}_{\min}$ vanishes under the substitution $x_1^{(i)} \to \mathcal{L}_{\bg + \varepsilon\ba(\bx)}^i(x_1)$ for $0 \leqslant i \leqslant n$.
    Since, for every $0 \leqslant i \leqslant n$, we have $\mathcal{L}_{\bg + \varepsilon\ba(\bx)}^i(x_1) = \mathcal{L}_{\bg}^i(x_1) + \mathrm{O}(\varepsilon)$, we deduce that $\mathcal{R}_{\bg}(f_0) = 0$.
    The minimality of $f_{\min}$ implies that $f_0$ is divisible by $f_{\min}$.
    Therefore, the Newton polytopes of $f_0$ and, thus, $\tilde{f}_{\min}$ contain a nonnegative shift of $\mathcal{N}$.
    Consider the vertices of the Newton polytope of $\tilde{f}_{\min}$.
    The corresponding coefficients are polynomials in $\ba$ and $\varepsilon$; we denote them by $p_1(\ba, \varepsilon), \ldots, p_{\ell}(\ba, \varepsilon) \in \KK[\ba, \varepsilon]$.
    We fix any nonzero $\varepsilon^\ast \in \KK$ such that none of $p_1(\ba, \varepsilon^\ast), \ldots, p_\ell(\ba, \varepsilon^\ast)$ is identically zero.
    We define an open subset $U_2 \subset V$ by 
    \[
    U_2 := \{\bg^\ast \in V \mid \forall \; 1 \leqslant i \leqslant \ell \colon p_i( (\bg^\ast - \bg) / \varepsilon^\ast , \varepsilon^\ast) \neq 0\}.
    \]
    Since polynomials $p_1(\ba, \varepsilon^\ast), \ldots, p_n(\ba, \varepsilon^\ast)$ are nonzero, we have $U_2 \neq \varnothing$.

    Let $\hat{f}_{\min}$ be the minimal polynomial for $x_1$ in $\bx' = \ba(\bx)$.
    Since the coefficients of the right-hand side of this system as well as of $\bx' = \bg + \varepsilon \ba(\bx)$ are algebraically independent over $\KK$, the degrees of $\tilde{f}_{\min}$ and $\hat{f}_{\min}$ coincide.
    We denote this degree by $\Rb \delta$.
    Let $\mathcal{M}$ be the collection of all the monomials in $x_1^{(\leqslant n)}$ of degree at most ${\Rb \delta} - 1$.
    Due to the minimality of $\hat{f}_{\min}$ and Lemma~\ref{lem::substitution}, polynomials $\{\mathcal{R}_{\ba}(m) \mid m \in \mathcal{M}\}$ are linearly independent over~$\KK(\ba)$.
    Consider these polynomials in the monomial basis and choose a nonsingular minor, we denote its determinant by $q(\ba)$.
    Let $U_3 \subset V$ be the open subset defined by $q(\ba) \neq 0$.

    Finally, we set $U := U_1 \cap U_2 \cap U_3$ and will prove that it satisfies the requirements of the proposition.
    Let $\bg^\ast \in U$, and we denote the minimal polynomial for $x_1$ in $\bx' = \bg^\ast(\bx)$ by $f_{\min}^\ast$.
    The fact that $\bg^\ast \in U_1$ implies that $f_{\min}^\ast$ is of order $n$.
    By setting $\ba^\ast(\bx) := (\bg^\ast - \bg)/ \varepsilon^\ast $, we see that the elimination ideal $I_{\bg^\ast} \cap \KK[x_1^{(\infty)}]$ contains $\tilde{f}_{\min}^\ast := \tilde{f}_{\min}|_{\ba = \ba^\ast}$.
    By the choice of $U_2$, the Newton polytope of $\tilde{f}_{\min}^\ast$ is the same as for $\tilde{f}_{\min}$.
    In particular, it contains {\Ra a} shift {\Ra of} $\mathcal{N}$, so it remains to show that $\tilde{f}_{\min}^\ast$ is the minimal polynomial of $I_{\bg^\ast} \cap \KK[x_1^{(\infty)}]$.
    To this end, we use that $\bg^\ast \in U_3$ which implies that $\{\mathcal{R}_{\bg^\ast}(m) \mid m \in \mathcal{M}\}$ are linearly independent over $\KK$, so there is no polynomial in $I_{\bg^\ast} \cap \KK[x_1^{(\infty)}]$ of degree less than ${\Rb \delta} = \deg \tilde{f}_{\min}^\ast$.
    So $\tilde{f}_{\min}^\ast$ is the minimal polynomial in $I_{\bg^\ast} \cap \KK[x_1^{(\infty)}]$.
\end{proof}


\subsection{Auxiliary statement about Vandermonde type systems}\label{sec:vandermonde}

In this section, we will prove a lemma which will be used in the next section to analyze the solutions at infinity of the polynomial systems obtained by taking Lie derivatives.
Throughout the section we will fix the field of Laurent series $\KK((\bz)) = \KK((z_1))((z_2))\ldots((z_n))$ with the lexicographic monomial ordering with
    \[
        z_n > z_{n-1} > \ldots > z_1.
    \] 

\begin{lemma} \label{lem::Puiseux}
    Let $n, d, D$ be positive integers and $\bgamma \in \mathbb{Z}_{\geqslant 0}^{n}$ such that $ \gamma_1 < \ldots < \gamma_n $.
    For $\balpha$ in $\overline{\KK}^n$ consider the square system:
    \begin{equation} \label{eq::Zariski_new}
         \begin{cases}
             \alpha_1^{\gamma_1} x_1^{d + \gamma_1(D-1)} + \cdots +  \alpha_n^{\gamma_1} x_n^{d + \gamma_1(D-1)} = 0,\\
             \vdots\\
            \alpha_1^{\gamma_{n}} x_1^{d + \gamma_{n}(D-1)} + \cdots   + \alpha_n^{\gamma_{n}} x_n^{d + \gamma_{n}(D-1)}  = 0.\\
         \end{cases}
     \end{equation}
     Then there exists a non-empty Zariski open subset $U \subset \overline{\KK}^n$ such that for every choice $\balpha \in U$ the system \eqref{eq::Zariski_new} has no nonzero solutions in $\overline{\KK}^n$.
\end{lemma}

The idea of the shorter proof given below was proposed to us by Joris van der Hoeven.

\begin{proof} 
    We will prove the statement of Lemma \ref{lem::Puiseux} for the field $\overline{\KK((\bz))}^n$ instead of $\overline{\KK}^n$ and then demonstrate the transition to the original formulation.
    The proof begins with constructing $\balpha \in \overline{\KK((\bz))}^n$ such the system \eqref{eq::Zariski_new} has no nonzero solutions in $\overline{\KK((\bz))}^n$. Next, we establish the existence of an open set $U \subset \overline{\KK((\bz))}^n$ such that for all $\balpha \in U$ the system \eqref{eq::Zariski_new} has no nonzero solutions. Finally, we demonstrate the transition to the original statement of the lemma.

    \emph{Step 1.} Consider $\balpha^\ast = \tuple{z_1, \ldots, z_n}$. 
    We show that for $\balpha^\ast$ the system \eqref{eq::Zariski_new} has no nonzero solutions in $\overline{\KK((\bz))}^n$ via induction on $n$. For the base case of induction at $n=1$ the system \eqref{eq::Zariski_new} takes the shape
    \begin{equation} \label{eq::Zariski_basecase}
            \alpha_1^{\gamma_1} x_1^{d + \gamma_1(D-1)} = 0.
    \end{equation}
    Note that for every $\alpha_1 \neq 0$, equation \eqref{eq::Zariski_basecase} has no nonzero solutions in $\overline{\KK((z_1))}$, so as for $\alpha_1 = z_1$.
    
    Now consider the system \eqref{eq::Zariski_new} for $n > 1$.
    Assume for contradiction that the system \eqref{eq::Zariski_new} has a nonzero solution $\tuple{x_1^\ast,\ldots, x_n^\ast}$ in $\overline{\KK((\bz))}^n$. 
    If $x_n^\ast = 0$, then the square system of the first $n-1$ equations of the system \eqref{eq::Zariski_new} has no nonzero solutions with $\balpha = \tuple{z_1, \ldots, z_{n-1}}$ by the induction hypothesis.  

    If $x_n^\ast \neq 0$, then we dehomogenize the system \eqref{eq::Zariski_new} with respect to $x_n$ (i.e., set $\tilde{x}_i = \frac{x_i^\ast}{x_n^\ast}$):
    \begin{equation} \label{eq::Zariski_dehomogenized}
         \begin{cases}
             z_1^{\gamma_1}\tilde{x}_1^{d+\gamma_1(D-1)} + \cdots + z_{n-1}^{\gamma_1}\tilde{x}_{n-1}^{d+\gamma_1(D-1)} + z_n^{\gamma_1} = 0,\\
             \vdots\\
            z_1^{\gamma_n} \tilde{x}_1^{d + \gamma_n(D-1)} + \cdots   + z_{n-1}^{\gamma_n} \tilde{x}_{n-1}^{d + \gamma_n(D-1)} + z_n^{\gamma_n}  = 0.\\
         \end{cases}
     \end{equation}
        
    Let $a$ be the minimum of orders of $\tilde{x}_i$ in $z_n$. For every $1 \leqslant i \leqslant n-1$ we express $\tilde{x}_i$ as a Puiseux series in $z_n$ with coefficients in $\overline{\KK((z_1))\ldots((z_{n-1}))}$
    \begin{equation} \label{eq::Zariski_substitution}
        \tilde{x}_i = \theta_i z_n^a + p_i(z_n),
    \end{equation}
    where $\theta_i \in \overline{\KK((z_1))\ldots((z_{n-1}))}$ and $p_i(z_n)$ is a Puiseux series in $z_n$ with coefficients in $\overline{\KK((z_1))\ldots((z_{n-1}))}$ such that $\deg p_i(z_n) > a$.
    Note that by the construction $\theta_i \neq 0$ for some $1 \leqslant i \leqslant n-1$.

    We substitute \eqref{eq::Zariski_substitution} to the system \eqref{eq::Zariski_dehomogenized}  and consider the lowest terms of the equations.
    Assume $a(d + \gamma_1(D-1)) > \gamma_1$. Then the lowest term of the first equation is $z_n^{\gamma_1}$ and it does not cancel. So we do not have the first equality.
    Thus, we can further assume that $a(d + \gamma_1(D-1)) \leqslant \gamma_1$. 
    Since $a \leqslant \frac{\gamma_1}{d + \gamma_1(D-1)}$ and $\gamma_i = \gamma_1 + N_i$ for some $N_i \in \mathbb{Z}_{>0}$ for every $2 \leqslant i \leqslant n$, we have 
    \[
        a(d + \gamma_i(D-1)) \leqslant \frac{\gamma_1( d + \gamma_i(D-1)}{d + \gamma_1(D-1))} = \gamma_1(1 + \frac{N_i(D-1)}{d + \gamma_1 (D-1)}) < \gamma_i.
    \]
    Thus, for the coefficients for the lowest terms of the last $n-1$ equations of the system \eqref{eq::Zariski_substitution} we have
    \begin{equation*}
        \begin{cases}
            z_1^{{\gamma}_2} \theta_1^{d + {\gamma}_2(D-1)} + \cdots + z_{n-1}^{{\gamma}_2} \theta_{n-1}^{d + {\gamma}_2(D-1)} = 0, \\
            \vdots \\
            z_1^{{\gamma}_{n}} \theta_1^{d + {\gamma}_{n}(D-1)} + \cdots + z_{n-1}^{{\gamma}_{n}} \theta_{n-1}^{d + {\gamma}_{n}(D-1)} = 0. \\
        \end{cases}
    \end{equation*}
    This system as a system in variables $\theta_1, \ldots, \theta_n$ has no nonzero solutions by the induction hypothesis, and so we get a contradiction with $\theta_i \neq 0$ for some $1 \leqslant i \leqslant n-1$. This finishes the first step of the proof.

    Since the substitution $\tuple{z_1, \ldots, z_n} \rightarrow \tuple{s_1 z_1, \ldots, s_n z_n}$, $s_i \in \mathbb{Z} \setminus \{ 0 \}$ to the system \eqref{eq::Zariski_new} does not change the lowest terms of equations \eqref{eq::Zariski_dehomogenized} in $z_n$, with the same argument as for $\balpha^\ast$ we can show that for $\balpha = \tuple{s_1 z_1, \ldots, s_n z_n}$ the system \eqref{eq::Zariski_new} has no nonzero solutions in $\overline{\KK((\bz))}^n$.

    \emph{Step 2.} 
    Consider the first-order formula in the language of fields 
    \[
        \Phi(\balpha) := \forall \bx \bigl((\bx \text{ is a solution of \eqref{eq::Zariski_new} with coefficients $\balpha$} ) \Rightarrow (x_1 = 0 \vee \ldots \vee x_n=0) \bigr)
    \]
    and let $E_0$ be the set $\{ \balpha \in \overline{\KK}^n \; | \; \Phi(\balpha)\}$.
    The set $E_0$ is constructible, so $E_0$ is contained in some proper Zariski closed subset or $E_0$ contains a nonempty Zariski open set \cite[Ch.~II, §~3, Ex.~3.18(b)]{hartshorne1977graduate}.
    
    Assume that $E_0$ is contained in some proper Zariski closed set defined by the polynomial $p(\balpha)$ over $\overline{\KK}$. Then the formula
    \begin{equation} \label{eq::formula}
        \forall \balpha \bigl(\Phi(\balpha) \Rightarrow (p(\balpha) = 0)\bigr)
    \end{equation}
    is true over $\overline{\KK}$. Since $\overline{\KK}$ and $\overline{\KK((\bz))}$ are elementary equivalent as algebraic closed fields of zero characteristic \cite[Th.~1.4, Th.~1.6]{marker2017model}, formula \eqref{eq::formula} is true over $\overline{\KK((\bz))}$. 
    On the other hand, as shown in Step 1, the set $E_1 := \{ \balpha \in \overline{\KK((\bz))}^n \; | \; \Phi(\balpha) \}$ contains a Zariski dense set $\{(s_1 z_1, \ldots, s_n z_n) \mid s_i \in \mathbb{Z} \setminus \{ 0\}\}$. 
    This leads to a contradiction with the assumption that $E_0$ was contained in a proper Zariski closed set.

    Thus, $E_0$ contains a nonempty Zariski open subset $U$, such that for every choice $\balpha \in U \subset \overline{\KK}^n$ the system \eqref{eq::Zariski_new} has no nonzero solutions in $\overline{\KK}^n$. This concludes the proof.
\end{proof}

\subsection{Proof of Theorem~\ref{thm:sharp}}\label{sec:proof_general_sharp}

\begin{notation}\label{not:lower_bound}
For this section, we will fix positive integers $d, D, n$ satisfying $d \leqslant D$.
For a tuple $\balpha = \tuple{\alpha_2, \ldots, \alpha_n} \in \KK^{n - 1}$, we consider an ODE system $\bx' = \bg_{\balpha}(\bx)$:
\begin{equation}\label{eq:sigma}
\begin{cases}
    x_1' = x_1^{d} + (x_2 + 1)^{d} + \cdots + (x_n + 1)^d,\\
    x_2' = \alpha_2 x_2^D,\\
    \vdots\\
    x_n' = \alpha_n x_n^D.
\end{cases}
\end{equation}
That is, $g_{\balpha, 1} = x_1^{d} + (x_2 + 1)^{d} + \cdots + (x_n + 1)^d$ and $g_{\balpha, i} = \alpha_i x_i^D$ for $2 \leqslant i \leqslant n$.

We also denote (see Notation~\ref{not:proofs_notation})
\begin{equation}\label{eq:pi}
  p_{\balpha, i} := x_1^{(i)} - \mathcal{L}_{\bg_{\balpha}}^{i - 1}( g_1 ) \quad\text{ for }1 \leqslant i \leqslant n.
\end{equation}
We further denote $I_{\balpha} := \ideal{p_{\balpha, 1}, \ldots, p_{\balpha, n}} \subseteq \mathbb{K}[x_1, \ldots, x_1^{(n)}, x_2, \ldots, x_n]$. 
Finally, we introduce the following fields of rational functions
\begin{equation}\label{eq:field}
F_\ell := \KK(x_1, \ldots, x_1^{(\ell - 1)}, x_1^{(\ell + 1)}, \ldots, x_1^{(n)}) \quad \text{ where }\quad 0 \leqslant \ell \leqslant n.
\end{equation}
A polynomial which reaches the bound given in Theorem~\ref{theorem_general_specialized}, by formula~\eqref{eq:bound_1}, must have the following degrees in $x_1, x_1', \ldots, x_1^{(n)}$: 
\[
N_0 := \prod\limits_{k = 1}^n (d + (k - 1)(D - 1)), \quad N_\ell := \frac{N_0}{d + (\ell - 1)(D - 1)} \text{ for } {\Rb \ell} = 1, \ldots, n.
\]

\end{notation}

\begin{lemma}\label{lem:sharp_radical}
    For every $0 \leqslant \ell \leqslant n $, the ideal generated by $p_{\balpha, 1}, \ldots, p_{\balpha, n}$ in $F_\ell[x_1^{(\ell)}, x_2, x_3, \ldots, x_n]$ is radical.
\end{lemma}

\begin{proof}  
    We fix a monomial ordering on $\KK[x_1^{(\leqslant n)}, x_2, \ldots, x_n]$ to be the lexicographic monomial ordering with 
    \[
    x_1^{(n)} > x_1^{(n-1)} > \ldots > x'_1 > x_1 > x_2 > \ldots > x_n.
    \]
    The leading term of $p_{\balpha, i}$ is $x_1^{(i)}$.
    Therefore, the leading terms of all generators of $I_{\balpha}$ are distinct variables. 
    Hence this set is a Gr\"obner basis of $I_{\balpha}$
    by the Buchberger's first criterion. 
    Then $\KK[x_1^{(\leqslant n)}, x_2, \ldots, x_n] / I_{\balpha} \simeq \mathbb{K}[\mathbf{x}]$ and the ideal $I_{\balpha}$ is prime and, thus, radical. 
    By \cite[Proposition 3.11]{atiyah2018introduction} the ideal generated by $I_{\balpha}$ in the localization $F_\ell[x_1^{(\ell)}, x_2, x_3, \ldots, x_n]$ will be radical.
\end{proof}

\begin{corollary} \label{cor::sharp_radical}
   For every $0 \leqslant \ell \leqslant n $, the ideal generated by $p_{\balpha, 1}, \ldots, p_{\balpha, \ell - 1},  p_{\balpha, \ell + 1}, \ldots, p_{\balpha, n}$ in $F_\ell[x_2, x_3, \ldots, x_n]$ is radical. 
\end{corollary}

\begin{proof}
    Let $J = \ideal{p_{\balpha, 1}, \ldots, p_{\balpha, \ell - 1},  p_{\balpha, \ell + 1}, \ldots, p_{\balpha, n}} \subset F_\ell[x_2, \ldots, x_n]$.
    Since $p_{\balpha, \ell}$ is linear in $x_1^{(\ell)}$, we have $J = I_{\balpha} \cap F_\ell[x_2, x_3, \ldots, x_n]$.
    Then Lemma~\ref{lem:sharp_radical} implies that $J$ is radical.
\end{proof}

\begin{lemma} \label{lem::jacobi}
    If all the coordinates of $\balpha$ are nonzero, then for every $0 \leqslant \ell \leqslant n$, the images of $\{ x_1, x_1', \ldots, x_1^{(\ell - 1)}, x_1^{(\ell + 1)}, \ldots, x_1^{(n)} \}$ are algebraically independent in $\mathbb{K}[x_1^{(\leqslant n)}, x_2, \ldots, x_n]/I_{\balpha} \cong \mathbb{K}[\mathbf{x}]$.
\end{lemma}
\begin{proof}
    Denote $R:=\mathbb{K}[x_1^{(\leqslant n)}, x_2, \ldots, x_n]/I_{\balpha}$.
    We will denote the images of $x_1^{(\leqslant n)}, x_2, \ldots, x_n$ in $R$ by the same letters.
    We have the following equalities modulo $I_{\balpha}$:
    \[
      x_1^{(i)} = \mathcal{L}_{\bg_{\balpha}}^i(x_1) =: h_i(\bx), \quad \text{ for } 0 \leqslant i \leqslant n.
    \]
    Let $J$ be the Jacobian determinant of $h_0, \ldots,h_{\ell - 1}, h_{\ell + 1}, \ldots,  h_n$  with respect to $\mathbf{x}$.
    Since $h_0 = x_1$, the determinant $J$ is equal to the Jacobian of $h_1, \ldots,h_{\ell - 1}, h_{\ell + 1}, \ldots,  h_n$ with respect to $x_2, \ldots, x_n$.
   
    Consider the matrix consisting of the leading terms of each entry of this 
    Jacobian with respect to the degree lexicographic ordering with
    \[
        {\Rb  x_2 > \cdots > x_n > x_1}.
    \]
    Since 
    \[
    \lt\left( \frac{\partial h_i}{\partial x_j} \right) = c_i \alpha_j^{i - 1} x_j^{d + (D - 1)(i - 1) - 1},\quad \text{ for } 1 \leqslant i \leqslant n, \; 2 \leqslant j \leqslant n,
    \]
    where $c_i =  d \prod_{j = 1}^{i - 1} \bigl(d + j (D - 1)\bigr)$, the corresponding determinant of the leading terms is equal to
    \begin{equation*} 
     J_{\lt} = \begin{vmatrix}  
       c_1 x_2 ^{d-1} &\ldots & c_1 x_n ^{d-1} \\
        c_2 \alpha_2 x_2^{d + (D-1) -1 } & \ldots &  c_2 \alpha_n  x_n^{d + (D-1) -1 } \\
         \vdots & \ddots & \vdots \\
         c_{\ell - 1} \alpha_2^{\ell - 2} x_2^{d + (\ell - 1)(D - 1) - 1} &\ldots &  c_{\ell - 1} \alpha_n^{\ell - 2} x_n^{d + (\ell - 1)(D - 1) - 1}\\  
         c_{\ell + 1} \alpha_2^{\ell } x_2^{d + (\ell + 1)(D - 1) - 1} &\ldots &  c_{\ell + 1} \alpha_n^{\ell } x_n^{d + (\ell + 1)(D - 1) - 1}\\
         \vdots & \ddots & \vdots \\
          c_n \alpha_2^{n - 2} x_2^{d + (n - 1)(D - 1) - 1} &\ldots &  c_n \alpha_n^{n - 2} x_n^{d + (n - 1)(D - 1) - 1}\\
        \end{vmatrix}.
       \end{equation*}
       The leading term of the determinant $J_{\lt}$ is the product over the antidiagonal 
       \[
        c x_2^{d + (n - 1)(D - 1)- 1} x_3^{d + (n - 2)(D - 1)- 1} {\Ra \cdots} x_{\ell}^{d + (\ell + 1) (D - 1) - 1} x_{\ell + 1}^{d + (\ell - 1) (D - 1) - 1} {\Ra \cdots} x_n^{d-1}
       \]
       for a nonzero constant $c$. 
       Since $\lt(J_{\lt}) \neq 0$, we have $\lt(J) = \lt(J_{\lt}) \neq 0$.
       The claimed algebraic independence then follows from the Jacobian criterion \cite[Theorem~2.2]{ehrenborg1993apolarity}.
\end{proof}

\begin{lemma}\label{lem:sharp:count}
    Following Notation \ref{not:lower_bound}, consider the system of equations
    \begin{equation} \label{eq::number_of_solutions_zero}
        p_{\balpha, 1} = \cdots = p_{\balpha, n} = 0.
    \end{equation}
    Then there exists a non-empty Zariski open subset $U\subset \mathbb{K}^{n-1}$ such that, for every $\balpha\in U$ and every $0 \leqslant \ell \leqslant n$, the system \eqref{eq::number_of_solutions_zero} has exactly $N_{\ell}$ distinct solutions in $\overline{F}_\ell$, the algebraic closure~of~$F_\ell$.
 \end{lemma}

We first prove an auxiliary lemma which we will use to rewrite the system~\eqref{eq::number_of_solutions_zero}.

\begin{lemma} \label{lem::diff_operator}
 Let $R$ be a differential $\mathbb{K}$-algebra and let $\mathcal{D}$ be a differential operator $\mathcal{D} : R \rightarrow R$.
   Fix $k\in \mathbb{Z}_{>0}$ and $b\in R$. 
   Assume that for some $a \in R$ we have $\mathcal{D}(a) = a^k + b$.
   Then the following equality of ideals in $R$ holds for every $n \geqslant {\Rb 2}$
\[
    \ideal{\mathcal{D}(a), \mathcal{D}^2(a), \ldots, \mathcal{D}^n(a)}  = \ideal{\mathcal{D}(a), \mathcal{D}(b), \mathcal{D}^2(b), \ldots, \mathcal{D}^{n-1}(b)}.
\]
\end{lemma}

\begin{proof}
    We will prove the equality of ideals by showing the inclusions in both directions {\Rb by} induction on $n$, starting with $n ={\Rb 2}$. 
    {\Rb For the base case $n = 2$: 
    \[ 
    \mathcal{D}^2(a) = \mathcal{D}(a^k + b) = k a^{k-1} \mathcal{D}(a) + \mathcal{D}(b) 
    \in \ideal{\mathcal{D}(a), \mathcal{D}(b)}.
    \]
    } 
    {\Rb Rearranging gives
    $\mathcal{D}(b) = \mathcal{D}^2(a) - k a^{k-1} \mathcal{D}(a) 
    \in \ideal{\mathcal{D}(a), \mathcal{D}^2(a)}$. 
    Thus, 
    $\ideal{\mathcal{D}(a), \mathcal{D}^2(a)} = \ideal{\mathcal{D}(a), \mathcal{D}(b)}$.}
    Now we consider the case $n \Rb > 2$. 
    {\Rb The induction hypothesis implies}
    \[
        \mathcal{D}^{n - 1}(a) \in \ideal{\mathcal{D}(a), \mathcal{D}(b), \ldots, \mathcal{D}^{n-2}(b)},
    \]
    then
    \[
        \mathcal{D}^{n}(a) \in \ideal{\mathcal{D}(a), \mathcal{D}^2(a), \mathcal{D}(b), \ldots, \mathcal{D}^{n - 1}(b)}.
    \]
    Since $\mathcal{D}^2(a) \in \ideal{\mathcal{D}(a), \mathcal{D}(b)}$, {\Rb by the base case,} we have 
    \[
        \mathcal{D}^{n}(a) \in \ideal{\mathcal{D}(a), \mathcal{D}(b), \ldots, \mathcal{D}^{n - 1}(b)}. 
    \]

    {\Rb In the other direction, using the induction hypothesis, we have}
    \[
    \mathcal{D}^{n-2}(b) \in \ideal{\mathcal{D}(a), \mathcal{D}^2(a), \ldots, \mathcal{D}^{n-1}(a)} {\Rb \implies \mathcal{D}^{n-1}(b) \in \ideal{\mathcal{D}(a), \mathcal{D}^2(a), \ldots, \mathcal{D}^{n}(a)}}\qedhere.
    \]
\end{proof}

\begin{proof}[Proof of Lemma~\ref{lem:sharp:count}]
    First of all we note that for every $1 \leqslant \ell \leqslant n$ the solutions of the system 
         \begin{equation} \label{eq::number_of_solutions_short}
         p_{\balpha, 1} = \cdots = p_{\balpha, \ell-1} = p_{\balpha, \ell+1} = \cdots = p_{\balpha, n} = 0     
         \end{equation} 
    in the variables $x_2, \ldots, x_n$ can be uniquely lifted to solutions of the system \eqref{eq::number_of_solutions_zero} in the variables $x_1^{(\ell)}, x_2, \ldots, x_n$ by substituting these solutions of \eqref{eq::number_of_solutions_short} to $p_{\balpha, \ell} = x_1^{(\ell)} - \mathcal{L}_{\bg_{\balpha}}^{{\Rb \ell - 1}}(g_{\balpha, 1}) = 0$. 
    Thus, the number of solutions of the system \eqref{eq::number_of_solutions_zero} for the case $1 \leqslant \ell \leqslant n$ is equal to the number of solutions of the system \eqref{eq::number_of_solutions_short}.
    
     \emph{Step 1: Zero-dimensionality.}
     Since $\mathbb{K}[x_1, \ldots, x_1^{(n)}, x_2, \ldots, x_n] / I_{\balpha} \simeq \mathbb{K}[\mathbf{x}]$, $I_{\balpha}$ is of dimension $n$. 
     Let $U_{-1} = (\KK \setminus \{ 0\})^{n-1} \subset \KK^{n-1}$ be a non-empty Zariski open set of all vectors without zero components. 
     For every $0 \leqslant \ell \leqslant n$ and $\balpha \in U_{-1}$ the set $x_1, \ldots, x_1^{(\ell - 1)}, x_1^{(\ell + 1)}, \ldots, x_1^{(n)}$ is algebraically independent modulo $I_{\balpha}$ by Lemma~\ref{lem::jacobi}.
     Since $\dim(I_{\balpha})=n$, it is a maximal algebraically independent set. 
     Thus, $I_{\balpha}\cdot F_{\ell}[x_1^{(\ell)}, x_2, \ldots, x_n]$ is a zero-dimensional ideal.

      Now for every $1 \leqslant \ell \leqslant n$ consider the ideal $J_\ell$ generated by $ p_{\balpha, 1}, \ldots, p_{\balpha, \ell-1}, p_{\balpha, \ell+1}, \ldots, p_{\balpha, n}$ in $\mathbb{K}[x_1, \ldots, x_1^{(\ell - 1)}, x_1^{(\ell + 1)}, \ldots, x_1^{(n)}, x_2, \ldots, x_n]$. 
      Since, among $p_{\balpha, 1}, \ldots,  p_{\balpha, n}$, only $p_{\balpha, \ell}$ involves the variable $x_1^{(\ell)}$ (and is linear in it), we have 
      \[
      I_{\balpha} \cap F_{\ell}[x_2, \ldots, x_n] = J_\ell \cdot F_{\ell}[x_2, \ldots, x_n].
      \]
      In particular, $\dim(J_\ell \cdot F_{\ell}[x_2, \ldots, x_n]) = 0$.
      
     \emph{Step 2: Solutions at infinity.}
     Consider the homogenizations $p_{\balpha, 1}^h, \ldots, p_{\balpha, n}^h$ of $p_{\balpha, 1}, \ldots, p_{\balpha, n}$ considered as polynomials in $x_1^{(\ell)}, x_2, \ldots, x_n$ using
     an additional variable $h$.

     For every $0 \leqslant \ell \leqslant n$, we will produce {\Rb a nonempty} Zariski open set $U_\ell \subset \KK^{n - 1}$ such that, for every $0 \leqslant \ell \leqslant n$ and for every $\balpha \in U_\ell$, the following system does not have a solution in $\overline{F}_\ell$ at infinity (that is, in $\mathbb{V}(h)$):
     \begin{itemize}
         \item if $\ell = 0$, then the system is $p_{\balpha, 1}^h = \cdots = p_{\balpha,n}^h = 0$ in variables $\bx, h$;
         \item if $1 \leqslant \ell \leqslant n$, then the system is $p_{\balpha, 1}^h = \cdots = p_{\balpha, \ell - 1}^h = p_{\balpha, \ell + 1}^h = \cdots = p_{\balpha,n}^h = 0$ in variables $x_2, \ldots, x_n, h$.
     \end{itemize}

     \emph{Assume that $d < D$.} 
     For $d < D$, the points at infinity for the case $\ell = 0$
     are given by the solutions (in projective space) of the equations $p_{\balpha,1}^h = \cdots = p_{\balpha,n}^h = h =0$. Here 
     \begin{equation} \label{eq::initial_condition}
         \begin{cases}
             p_{\balpha,1}^h|_{h=0} = x_1^d + x_2^d + \cdots + x_n^{d} = 0,\\
             p_{\balpha,2}^h|_{h=0} = c_1 ( \alpha_2 x_2^{d + D - 1} + \cdots + \alpha_n x_n^{d + D - 1}) = 0, \\
             \vdots\\
             p_{\balpha,n}^h|_{h=0} = c_n ( \alpha_2^{n-1} x_2^{d + (n-1)(D-1)} + \cdots +  \alpha_n^{n-1} x_n^{d + (n-1)(D-1)} ) = 0,\\
         \end{cases}
     \end{equation}
     where $c_i = \prod_{k=1}^{i}(d + (k-1)(D-1))$ for every $1 \leqslant i \leqslant n$.
     Note that the solutions of the $n - 1$ last equations of the system \eqref{eq::initial_condition} can be lifted to solutions of the entire system. 
     By applying Lemma~\ref{lem::Puiseux} for {\Rb the } $n - 1$ last equations of the system \eqref{eq::initial_condition} there exist a non-empty Zariski open subset $U_0 \subset \mathbb{K}^{n-1}$ such that for every choice $\balpha \in U_{0}$ the system \eqref{eq::initial_condition} has no nonzero solutions. 
     Therefore the system $p_{\balpha,1}^h = \ldots = p_{\balpha,n}^h = h =0$ has no nonzero solutions in $\overline{F}_0$ for $\balpha \in U_{0}$.

    For the case $1 \leqslant \ell \leqslant n$ 
    the points at infinity are the solutions of the system obtained from \eqref{eq::initial_condition} by substituting $x_1 = 0$ and omitting the $\ell$-th equation. 
    Then by applying Lemma~\ref{lem::Puiseux} to the resulting system \eqref{eq::initial_condition} {\Ra there exists} a non-empty Zariski open subset $U_{\ell} \subset \mathbb{K}^{n-1}$ such that for every choice $\balpha \in U_{\ell}$ the system \eqref{eq::initial_condition} has no nonzero solutions. 
    Therefore the system $p_{\balpha, 1}^h = \cdots = p_{\balpha, \ell - 1}^h = p_{\balpha, \ell + 1}^h = \cdots = p_{\balpha,n}^h = h = 0$ has no nonzero solutions in $\overline{F}_\ell$ for $\balpha \in U_{\ell}$.

    \emph{Now consider the case $d = D$.} 
    In this case the points at infinity for $\ell= 0$ are also given by the solutions in $\overline{F}_0$ of the equations $p_{\balpha,1}^h|_{h = 0} = \cdots = p_{\balpha,n}^h|_{h = 0} = 0$, where $p_{\balpha,1}^h|_{h = 0} = x_1^D + x_2^D + \cdots + x_n^D$. 
    Consider $p_{\balpha,1}^h|_{h = 0}$ as the image $\mathcal{D}(x_1)$ where $\mathcal{D}$ is defined to be a differential operator $\mathcal{D}\colon \mathbb{K}[\mathbf{x}] \rightarrow \mathbb{K}[\mathbf{x}]$ with 
    \[
      x_i \mapsto  (\mathcal{L}_{\bg_{\balpha}}(x_i))^h|_{h = 0},
    \]
   where $\mathcal{L}_{\bg}$ is the Lie derivative operator (see Notation~\eqref{not:proofs_notation}).
    Then we can rewrite the system $p_{\balpha,1}^h|_{h = 0} = \cdots = p_{\balpha,n}^h|_{h = 0} = h = 0$ as 
    \[
        \mathcal{D}(x_1) =  \mathcal{D}^2(x_1) = \mathcal{D}^3(x_1) = \cdots =  \mathcal{D}^{n}(x_1) = h = 0.
    \]
    By Lemma~\ref{lem::diff_operator} with $a = x_1, k = D$ and $b = x_2^D + \cdots + x_n^D$ this is equivalent to 
    \[
        \mathcal{D}(x_1) = \mathcal{D}(b) = \mathcal{D}^2(b) = \cdots = \mathcal{D}^{n-1}(b) = h = 0, 
    \]
    which itself is equivalent to the system~\eqref{eq::initial_condition} with $d = D$.
    Note that for the case $1 \leqslant \ell \leqslant n$ the points at infinity are given by \eqref{eq::initial_condition} with $d = D$ and $x_1 = 0$ and the $\ell$-th equation omitted. 
    Therefore, the same argument using Lemma \ref{lem::Puiseux} as in the case $d < D$ applies here.

    \emph{Step 3: Computing the B\'ezout bound.}
        Consider a non-empty open set $U := \bigcap\limits_{i=-1}^{n} U_{i} \subset \KK^{n-1}$.
        Thus, for $\ell = 0$ the ideal generated by $p_{\balpha,1}, \ldots, p_{\balpha,n}$ in $F_{\ell}[x_1, x_2, \cdots, x_n]$ is a zero-dimensional radical ideal by Lemma~\ref{lem:sharp_radical} and for every choice $\balpha \in U$ the system $p_{\balpha,1} = \ldots = p_{\balpha,n} = 0$
        has no solutions at infinity, so the number of distinct solutions, counted with multiplicity, 
        of the system \eqref{eq::number_of_solutions_zero} with $\balpha \in U$ is equal to 
        the B\'ezout bound. Therefore, {\Rb the number of solutions in $\overline{F}_{0}$}, is equal to the product of the total degrees of $p_{\balpha,1}, \ldots, p_{\balpha,n}$ in $\mathbf{x}$ which is equal to $N_{0}$. 
        
        For every $1 \leqslant \ell \leqslant n$ the ideal generated by $p_{\balpha,1},  \ldots, p_{\balpha, \ell - 1}, p_{\balpha, \ell + 1}, \ldots, p_{\balpha,n}$ is a zero-dimensional radical ideal in $\mathbb{F}_{\ell}[x_2, \ldots, x_n]$ by Corollary~\ref{cor::sharp_radical} and for every choice $\balpha \in U$ the system \eqref{eq::number_of_solutions_short}
        has no solutions at infinity, so the number of distinct solutions,  counted with multiplicity, of the system \eqref{eq::number_of_solutions_short} (considered in the variables $x_2, \ldots, x_n$), 
        as well as {\Rb the} number of {\Rb distinct solutions} of the system 
        \eqref{eq::number_of_solutions_zero} (considered in the variables $x_1^{(\ell)}, x_2, \ldots, x_n$),
        is equal to the B\'ezout bound for the system~\eqref{eq::number_of_solutions_short}, which is precisely~$N_{\ell}$.  
\end{proof}

\begin{lemma}\label{lem:distinct}
    Consider $0 \leqslant \ell \leqslant n-1$ and the system $p_{\balpha, 1} = \ldots = p_{\balpha, n} = 0$ (see Notation~\ref{not:lower_bound}) as a polynomial system in variables $x_1^{(\ell)}, x_2, \ldots, x_n$ over $F_\ell$.
    Assume that the coordinates of $\balpha$ are distinct prime numbers larger than $d$.
    Then the $x_1^{(\ell)}$-coordinates of the solutions of the system in $\overline{F}_\ell$ are all distinct.
\end{lemma}

The proof of Lemma~\ref{lem:distinct} will use the concept of \emph{identifiability} from control theory. 
Here we give a specialization of the general analytic definition~\cite[Definiton~2.5]{hong2020global} to the class of systems we consider.
\begin{defin}\label{def:identifiability}
    Let $\bx' = \bg(\bx)$ be a polynomial ODE system, where $\bx = \tuple{x_1, \ldots, x_n}$ and $\bg \in \mathbb{C}[\bx]^n$.
    Then the initial condition $x_i(0)$ is said to be \emph{identifiable} from $x_j$ if there exists {\Rb a nonempty} Zariski open $U \subset \mathbb{C}^n$ such that, for every solution $\mathbf{X}(t)$ of the system analytic at $t = 0$ with $\mathbf{X}(0) \in U$ and any other solution $\widetilde{\mathbf{X}}(t)$ analytic at $t = 0$ the equality $X_j(t) = \widetilde{X}_j(t)$ in a neighbourhood of $t = 0$ implies $X_i(0) = \widetilde{X}_i(0)$.
\end{defin}

The convenience of the notion of identifiability for us is that, while this property can be established by analytic means, it allows to deduce purely algebraic consequences.

\begin{lemma}\label{lem:identifiability_to_algebra}
    In the notation of Definition~\ref{def:identifiability}, if $x_i(0)$ is identifiable from $x_j$, then there exist polynomials $q, r \in \mathbb{C}[x_j^{(\leqslant n)}]$ such that $q\not\in I_{\bg}$ and $q x_i - r \in I_{\bg}$.
\end{lemma}

\begin{proof}
    By~\cite[Proposition~3.4]{hong2020global}, identifiability of $x_i(0)$ from $x_j$ implies that there exist $q, r \in \mathbb{C}[x_j^{(\infty)}]$ such that $q \not\in I_{\bg}$ and $qx_i - r \in I_{\bg}$.
    Let $f$ be the minimal polynomial of the elimination ideal $I_{\bg} \cap \mathbb{C}[x_j^{(\infty)}]$.
    We have $\ord_{x_j}f \leqslant n$ by \cite[Theorem~3.16 and Corollary~3.21]{hong2020global}.
    Therefore, by performing Ritt's reduction~\cite[Section~3.1]{Boulier2} of $qx_i {\Rb -} r$ with respect to $f$, we find $\tilde{q}, \tilde{r} \in \mathbb{C}[x_j^{(\leqslant n)}]$ such that $\tilde{q} \not\in I_{\bg}$ and $\tilde{q}x_i - \tilde{r} \in I_{\bg}$.
\end{proof}

It turns out that our system~\eqref{eq:sigma} does possess this property.

\begin{lemma}\label{lem:identifiability}
    Assume that the coordinates of $\balpha\in \mathbb{Q}^{n - 1}$ are distinct prime numbers greater than $d$.
    Then, in the system $\bx' = \bg
_{\balpha}(\bx)$ from~\eqref{eq:sigma}, considered over $\KK = \mathbb{C}$, the initial conditions $x_2(0), \ldots, x_n(0)$ are identifiable from~$x_1$.
\end{lemma}

\begin{proof}
    Due to the symmetry, it is sufficient to prove the identifiability of $x_2(0)$.
    To this end, we first observe that the complex-valued solutions of an equation $x' = \alpha x^D$ with $\alpha \neq 0$ analytic at $t = 0$ are the following:
    \begin{equation}\label{eq:analytic_solution}
    X(t) = \begin{cases}
        X(0) e^{{\Rb \alpha t}}, \text{ if } D = 1,\\
        (\alpha (1 - D) t + X(0)^{1 - D})^{\frac{1}{1 - D}}, \text{ if } D > 1 \text{ and } x(0) \neq 0,\\
        0, \text{ if } X(0) = 0.
    \end{cases}
    \end{equation}
    In order to verify Definition~\ref{def:identifiability}, we set the open $U \subset \mathbb{C}^n$ to be a set of vectors $\tuple{v_1, \ldots, v_n}$ with nonzero coordinates such that the numbers $\frac{v_2^{1 - D}}{\alpha_2}, \ldots, \frac{v_n^{1 - D}}{\alpha_n}$ are pairwise distinct.
    Consider any solutions $\bX(t)$ and $\widetilde{\bX}(t)$ of {\Rb the system}  $\Rb \bx' = \bg
_{\balpha}(\bx)$ in the ring of $\mathbb{C}$-valued functions locally analytic at $t = 0$ such that $\bX(0) \in U$ and $X_1(t) = \widetilde{X}_1(t)$ in a neighbourhood of $t  = 0$.
    Then we have $X_1'(t) - X_1(t)^d = \widetilde{X}_1'(t) - \widetilde{X}_1(t)^d$, so
    \begin{equation}\label{eq:equalityX_Xt}
    (X_2(t) + 1)^d + \ldots + (X_n(t) + 1)^d = (\widetilde{X}_2(t) + 1)^d + \ldots + (\widetilde{X}_n(t) + 1)^d.
    \end{equation}
    We denote $F(t) := (X_2(t) + 1)^d + \ldots + (X_n(t) + 1)^d$ and consider the cases $D = 1$ and $D > 1$ separately.
    
    \emph{Case $D = 1$.} Then $F(t)$ is a linear combination of exponential functions. 
    Since $\alpha_i > d$, among all the brackets in~\eqref{eq:equalityX_Xt}, the term $e^{\alpha_2t}$ can occur only from $(X_2(t) + 1)^d$ and $(\widetilde{X}_2(t) + 1)^d$.
    The coefficient in {\Rb the front of this term} will be equal to $d X_2(0)$ and $d \widetilde{X}_2(0)$, respectively.
    Due to the linear independence of the exponential functions with different growth rates these coefficients must be equal, so $X_2(0) = \widetilde{X}_2(0)$ as desired.

    \emph{Case $D > 1$.}
    By~\eqref{eq:analytic_solution}, the analytic continuation of $F(t)$ is a {\Rb multivalued} function with the branching points exactly at $-\frac{X_2(0)^{1 - D}}{\alpha_2 (1 - D)}, \ldots, -\frac{X_n(0)^{1 - D}}{\alpha_n (1 - D)}$, and, by the choice of $U$, these points are distinct.
    Therefore, $-\frac{\widetilde{X}_2(0)^{1 - D}}{\alpha_2 (1 - D)}, \ldots, -\frac{\widetilde{X}_n(0)^{1 - D}}{\alpha_n (1 - D)}$ must be a permutation of this set of points.
    Thus, these exists a unique $2 \leqslant j \leqslant n$ such that $C:= -\frac{X_2(0)^{1 - D}}{\alpha_2(1 - D)} = -\frac{\widetilde{X}_j(0)^{1 - D}}{\alpha_j (1 - D)}$.
    Then the principal (i.e., terms with negative degrees) part of the Puiseux expansion of $F(t)$ at $t = C$ will be equal to $(X_2(t) + 1)^d - 1$ on one hand and to $(\widetilde{X}_j(t) + 1)^d - 1$ on the other.
    This implies
    \[
    (X_2(t) + 1)^d = (\widetilde{X}_j(t) + 1)^d \implies X_2(t) = \omega \widetilde{X}_j(t) + (1 - \omega)
    \]
    for some $d$-th root of unity $\omega$.
    Since $\omega \widetilde{X}_j(t) + (1 - \omega)$ can be a solution of $x' = \alpha_2 x_j^D$ only for $\omega = 1$ and $j = 2$, we conclude that $X_2(t) = \widetilde{X}_2(t)$.
\end{proof}

\begin{proof}[Proof of Lemma~\ref{lem:distinct}]
    While we allow arbitrary $\KK$ of zero characteristic, the polynomial system is in fact defined over $\mathbb{Q}(x_1, \ldots, x_1^{(i - 1)}, x_1^{(i + 1)}, \ldots, x_1^{(n)})$, and the solutions will belong to the algebraic closure of this field.
    Therefore, proving the lemma for any fixed field $\KK$ would prove it for all fields, and we will choose $\KK = \mathbb{C}$.

    Combining Lemmas~\ref{lem:identifiability_to_algebra} and~\ref{lem:identifiability}, we conclude {\Rb that} there exist $q_2, \ldots, q_n, r_2, \ldots, r_n \in \mathbb{C}[x_1^{(\leqslant n)}]$ such that none of $q_2, \ldots, q_n$ belongs to $I_{\balpha}$ and $q_i x_i - r_i \in I_{\balpha}$ for every $2 \leqslant i \leqslant n$.

    Consider any $0 \leqslant \ell \leqslant n - 1$.
    Let $\tuple{a, \hat{x}_2, \ldots, \hat{x}_n}$ and $\tuple{a, \tilde{x}_2, \ldots, \tilde{x}_n}$ be two distinct solutions of $p_{\balpha, 1} = \cdots = p_{\balpha, n} = 0$ as polynomials in $x_1^{(\ell)}, x_2, \ldots, x_n$ over $\overline{F_\ell}$ with coinciding $x_1^{(\ell)}$-coordinate.
    Since the solutions are distinct, there exists  $2 \leqslant j \leqslant n$ with $\hat{x}_j \neq \tilde{x}_j$.
    {\Rb We define $b := q_j|_{x_1^{(\ell)} = a}$, and proceed to show that $b \neq 0$}. Assume for the contradiction that $b = 0$.
    Then $q_j$ and $f_{\min}$ of $I_{\balpha}$ have a common root as polynomials in $F_\ell[x_1^{(\ell)}]$.
    Since $f_{\min}$ is irreducible, it would imply that $q_j$ is divisible by $f_{\min}$ which is impossible due to $q_j \not\in I_{\balpha}$.
    Therefore, plugging our two solutions in $\Rb q_j x_j - r_j$ and using the fact that $\ideal{p_{\balpha, 1}, \ldots, p_{\balpha, n}} = I_{\balpha} \cap \mathbb{C}[x_1^{(\leqslant n)}, x_2, \ldots, x_n]$ (Lemma~\ref{lem:ideal_is_wanted}), we have
    \[
    0 = b \hat{x}_j {\Rb - } r_j|_{x_1^{(\ell)} = a} = b\tilde{x}_j {\Rb - } r_j|_{x_1^{(\ell)} = a}.
    \]
    Together with $b \neq 0$, this implies $\hat{x}_j = \tilde{x}_j$ leading to a contradiction with the existence of distinct solutions with the same $x_1^{(\ell)}$-coordinate.
\end{proof}

We can now combine the established properties of $\bx' = \bg_{\balpha}(\bx)$ from~\eqref{eq:sigma} for proving the sharpness of our bound.

\begin{proof}[Proof of Theorem~\ref{thm:sharp}]
   Let $U$ be the Zariski open set from Lemma~\ref{lem:sharp:count}.
   Since {\Ra the set of tuples of distinct prime numbers is Zariski dense in $\mathbb{Q}^{n-1}$},
   we can choose $\balpha \in \mathbb{Q}^{n - 1}$ such that $\balpha \in U$ and the coordinates of $\balpha$ are prime numbers greater than $d$.
   We will prove that the bound from Theorem~\ref{theorem_general_specialized} is achieved on $\bx' = \bg_{\balpha}(\bx)$.

    Consider $f_{\min}$ for the elimination ideal $I_{\alpha} \cap \mathbb{K}[x_1^{(\infty)}]$.
    By Lemma~\ref{lem:ideal_is_wanted}, $f_{\min}$ belongs to the ideal generated by $p_{\balpha, 1}, \ldots, p_{\balpha, n}$.
    Fix any $0 \leqslant \ell \leqslant n$.
    If we consider $f_{\min}$ as a polynomial in $F_{\ell}[x_1^{(\ell)}]$ (see~\eqref{eq:field}), then it must vanish on the roots of the system $p_{\balpha, 1} = \cdots = p_{\balpha, n} = 0$ in $\overline{F}_\ell$.
    Therefore, its degree in $\Rb x_1^{(\ell)}$ must be greater or equal than the number of distinct $x_1^{(\ell)}$-coordinates of the solutions of the system.
    Lemmas~\ref{lem:sharp:count} and~\ref{lem:distinct} imply that this number is equal to $N_\ell$.
    Thus, $\deg_{x_1^{(\ell)}} f_{\min} \geqslant N_\ell$.
    On the other hand, Theorem~\ref{theorem_general_specialized} implies that the only monomial of degree $N_\ell$ in $x_1^{(\ell)}$ which can appear in $f_{\min}$ is $(x_1^{(\ell)})^{N_\ell}$.
    Therefore, this monomial does appear with a nonzero coefficient.

    For $\varepsilon \in \KK$, we define $\bg_{\balpha, \varepsilon}$ as the result of applying a transformation $x_1 \to x_1 + \varepsilon$ to $\bg_{\balpha}$.
    Since this transformation is invertible, it maps the minimal polynomial of $\bg_{\balpha}$ to the minimal polynomial of $\bg_{\balpha, \varepsilon}$.
    That is, the latter is equal to $f_{\min, \varepsilon} := f_{\min}(x_1 + \varepsilon, x_1', \ldots, x_1^{(n)})$.
    We have proved that $f_{\min}$ contains a monomial $x_1^{N_0}$.
    Therefore, the constant term of $f_{\min, \varepsilon}$ considered as a polynomial in $x_1^{(\infty)}$ is a nonzero polynomial in $\varepsilon$.
    Thus, there exists $\varepsilon^\ast \in \KK$ such that $f_{\min, \varepsilon^\ast}$ has a nonzero constant term.
    So the Newton polytope of $f_{\min, \varepsilon^\ast}$ is exactly the simplex defined by~\eqref{eq:bound_1}.
    This already shows that the bound~\eqref{eq:bound_1} is sharp.

    In order to prove the generic sharpness as stated in Theorem~\ref{thm:sharp}, we will apply Proposition~\ref{prop:specific_to_generic}.
    We will take $\mathcal{P}_1$ (resp., $\mathcal{P}_2, \ldots, \mathcal{P}_n$) to be the simplex containing all the points with the sum of the coordinates not exceeding $d$ (resp., $D$).
    Then $V(\mathcal{P}_1)$ (resp., $\mathcal{P}_2, \ldots, \mathcal{P}_n$) will be equal to $V_d$ (resp., $V_D$) in the notation of Theorem~\ref{thm:sharp}.
    We also take $\bg^\ast = \bg_{\balpha, \varepsilon^\ast}$, and denote the Newton polytope of its minimal polynomial by $\mathcal{N}^\ast$.
    Then by Proposition~\ref{prop:specific_to_generic} there exists $U \subset V_d \times V_D^{n - 1}$ such that, for every $\bg \in U$, the Newton polytope $\mathcal{N}$ of the minimal polynomial for $x_1$ contains a shift of $\mathcal{N}^\ast$.
    On the other hand, by Theorem~\ref{theorem_general_specialized}, we have $\mathcal{N} \subset \mathcal{N}^\ast$, so $\mathcal{N} = \mathcal{N}^\ast$.
\end{proof}


\subsection{Proof of Theorem~\ref{thm:2d}} \label{sec:planar_sharpness}

We state and prove the following auxiliary lemma in full generality but we will use it only in the planar case.

\begin{lemma} \label{lem03}
    Let $p_1, \ldots, p_n$ be polynomials of degrees $d_1, \ldots, d_n$ in $\mathbb{K}[\mathbf{x}, \mathbf{y}]$ where $\bx = \tuple{x_1, \ldots, x_{n - 1}}$ and $\by = \tuple{y_1, \ldots, y_k}$.
    Let $\overline{p}_i$ denote the homogeneous component of degree $d_i$ in $p_i$, and assume that  $\overline{p}_i \in \mathbb{K}[\mathbf{x}]$ for every $1 \leqslant i \leqslant n$.
    Suppose that the ideal $I = \ideal{p_1, \ldots, p_n} \cap \mathbb{K}[\mathbf{y}]$ is principal, that is,
    $I = \ideal{g}$,
    {\Rb and assume that its generator} $g = g(\mathbf{y})$ {\Rb is} a nonzero irreducible polynomial. Then 
    $$\Res_{\mathbf{x}}(p_1, \ldots, p_n) = c\cdot g^m, \text{ for some } c \in \mathbb{K}, m \in \mathbb{Z}_{> 0}.$$
    \end{lemma}
    \begin{proof}

    We homogenize $p_1, \ldots, p_n$ in $\mathbf{x}$ using an additional variable $z$ and obtain
    $p_1^h, \ldots, p_n^h$. 

    Consider $\beta = \tuple{\beta_1, \ldots, \beta_k} \in \overline{{\mathbb{K}}}^k$.
    By~\cite[Ch.3~Th.~2.3]{cox2005using} $\beta$ is a zero of $\Res_{\mathbf{x}, z}(p_1^h, \ldots, p_n^h)$ if and only if there exists $\alpha = [\alpha_1: \ldots: \alpha_n] \in \mathbb{P}^{n - 1}$ such that $(\alpha, \beta)$ is a common zero of $p_1^h, \ldots, p_n^h$.

    If $\alpha_n = 0$, then $[\alpha_1 : \ldots : \alpha_{n - 1}]$ is a common zero for $\overline{p}_1, \ldots, \overline{p}_n$.
    Since $\overline{p}_1, \ldots, \overline{p}_n$ do not depend on $\mathbf{y}$, $(\alpha, \gamma)$ is a common root of $p_1^h, \ldots, p_n^h$ for every $\gamma \in \overline{\mathbb{K}}^k$.
    Therefore, $\Res_{\mathbf{x}, z} (p_1^h, \ldots, p_n^h)$ is identically zero, so is $\Res_{\mathbf{x}}(p_1, \ldots, p_n)$.
    Thus, we can take $c = 0$ and $m = 1$.
    Therefore, {\Rb for the rest of the proof we will assume that
     $\overline{p}_1, \ldots, \overline{p}_n$
    have no common zero in $\mathbb{P}^{n-1}$, which means that any common zero 
    $(\alpha, \beta)$ of $p_1^h, \ldots, p_n^h$ in $\mathbb{P}^{n-1} \times \overline{\mathbb{K}}^k$ 
    must satisfy $\alpha_n \neq 0$.}
    
    
    Since $\alpha_n \neq 0$, then $\tuple{\dfrac{\alpha_1}{\alpha_n}, \ldots, \dfrac{\alpha_{n-1}}{\alpha_n}, \beta}$ is a common zero of $p_1, \ldots, p_n$.
    Thus, $\beta$ belongs to the projection of the solution set of $p_1 = \cdots = p_n = 0$.
    Therefore, the zero set of $\Res_{\mathbf{x}}(p_1, \ldots, p_n)$ coincides with this projection.
    
    By the elimination theorem \cite[Ch.~3, §~2, Th.~2]{cox1997ideals}, the zero set of $I$ is the closure of the projection, so, by the Hilbert's Nullstellensatz, $g^{m_1} \in \ideal{\Res_{\mathbf{x}}(p_1, \ldots, p_n)}$ for some $m_1 \in \mathbb{Z}_{> 0}$. 
    Hence, there exists $q \in \mathbb{K}[\mathbf{y}]$ such that

    \[
    g^{m_1} = q \cdot \Res_{\mathbf{x}}(p_1, \ldots, p_n).
    \]
    Since the factorization of $g^{m_1}$ is unique up to multiplication of the factors by invertible constants and $g$ is a nonzero irreducible polynomial, then $q = \Tilde{c} \cdot g^{m_2}$ for some $\Tilde{c} \in \mathbb{K}\backslash \{0\}, m_2 \in \mathbb{Z}_{\geqslant 0}$, and thus for $ m = m_1 - m_2, \; c = \frac{1}{\Tilde{c}}$ we have

    \[
    \Res_{\mathbf{x}}(p_1, \ldots, p_n) = c \cdot g^{m}.\qedhere
    \]
    \end{proof}

\begin{lemma} \label{lem5}
For the following polynomials $\bg = \tuple{g_1, g_2}$ in $\mathbb{K}[x_1, x_2] = \mathbb{K}[\mathbf{x}]$ of degrees $d_1, d_2 > 0$
    \begin{equation*}
            g_1(x_1, x_2)  = x_1^{d_1} + x_2^{d_1},\quad \text{ and }\quad
            g_2(x_1, x_2)  = x_2^{d_2} + 1,
    \end{equation*}
 the minimal polynomial $f_{\min}$ of $I_{\bg} \cap \KK[x_1^{(\infty)}]$ contains monomials $x_1^{d_1(d_1+d_2-1)}$ and $(x_1'')^{d_1}$.  
If $d_1 > d_2$ then $f_{\min}$ also contains the monomial $x_1^{d_1(d_1-1)}(x'_1)^{d_1}$. 
\end{lemma}
\begin{proof}
    Let $p_1 := x_1' - g_1(x_1, x_2)$ and
    \[
      p_2 := \mathcal{L}_{\bg}^\ast (p_1) = x''_1  -  d_1( x_1^{d_1-1} x'_1 +  x_2^{d_1+d_2-1} +   x_2^{d_1 - 1}).
    \]
    We compute $\Res_{x_2}(p_1, p_2)$. 
    Since $p_1$, as a polynomial in the variable $x_2$, has the roots $\alpha_i$, $1 \leqslant i \leqslant d_1$, where $\alpha_i = \xi_i (x'_1 - x_1^{d_1})^{\frac{1}{d_1}}$, $\{1^{\frac{1}{d_1}} \} = \{\xi_1, \ldots, \xi_{d_1}\}$,  
      then, by the Poisson formula~\cite[p.~427]{gelfand1994discriminants}, we have  \begin{equation}\label{eq:poisson}
        \Res_{x_2}(p_1, p_2) = \prod_{i=1}^{d_1} p_2 (\alpha_i) = \prod_{i = 1}^{d_1}(x''_1 - d_1 b(\alpha_i)), \; \text{ where } b(t) := x_1^{d_1 - 1}x'_1 + t^{d_1 + d_2 -1} + t^{d_1 - 1}. 
    \end{equation}
    By expanding the brackets, we obtain
 \begin{equation}\label{eq:poisson_expanded}
            \Res_{x_2}(p_1, p_2) = (x''_1)^{d_1} + (-d_1)^{d_1}\prod_{i=1}^{d_1} b(\alpha_i) + x''_1 p(x_1, x'_1, x''_1),
    \end{equation} 
    for some polynomial $p \in \overline{\mathbb K(x_1, x'_1)}[x''_1]$ with $\deg_{x_1''} p < d_1 - 1$. 
    Thus, $\Res_{x_2}(p_1, p_2)$ contains the monomial $(x''_1)^{d_1}$.

    We define $\tilde{p}_i := p_i|_{x_1' = \Rb x_1'' = 0}$ for $i = 1, 2$.
    Since $p_1$ and $p_2$ are $x_2-$monic and $\deg_{x_2} \tilde{p}_i = \deg_{x_2} p_i$, we have
     \[
    \Res_{x_2}(\tilde{p}_1,\tilde{p}_2) = \Res_{x_2}(p_1, p_2)|_{x'_1= x''_1 = 0}.
    \]
    Then
    \[
     \Res_{x_2}(\Tilde{p}_1,\Tilde{p}_2) = {\Rb{(- d_1)^{d_1}}} \prod_{i=1}^{d_1} 
     ({\Rb -\xi_i^{d_2-1} (-1)^{\frac{d_2 - 1}{d_1}}} x_1^{d_1 + d_2 - 1} + {\Rb\xi_i^{-1} (-1)^{\frac{d_1-1}{d_1}}} x_1^{d_1 - 1}),
    \]
    and this polynomial reaches the highest degree only in the term $x_1^{d_1(d_1 + d_2 - 1)}$. Then $\Rb \Res_{x_2}(\Tilde{p}_1,\Tilde{p}_2)$ contains the monomial $x_1^{d_1(d_1 + d_2 - 1)}$, so $\Res_{x_2}(p_1,p_2)$ does.

    Assume $d_1 > d_2$ and denote $(x'_1 - x_1^{d_1})^{\frac{1}{d_1}}$ by $a$.
    Consider the ring 
     $\mathbb{K}[x_1, x'_1, \Rb x''_1]$ with respect to the grading $\wdeg x''_1 = 0$, $\wdeg x'_1 = d_1$, $\wdeg x_1 = 1$.
    Since $x_1' - x_1^{d_1}$ is homogeneous of degree $d_1$ with respect to this grading, we can extend the grading to the ring $\mathbb K [x_1, x'_1, x''_1, a]$ by setting $\wdeg a = 1$.
    Then the expression
    \[
    b(\alpha_i) = x_1^{d_1 - 1} x'_1 + \xi_i^{d_2 - 1} a^{d_1 + d_2 - 1} {\Rb + } \xi_i^{-1} a^{d_1 -1}
    \]
    reaches the highest degree only in the term $x_1^{d_1 - 1} x'_1$. 
    Therefore, $\prod\limits_{i=1}^{d_1} b(\alpha_i)$ contains the monomial $\Rb x_1^{d_1(d_1 - 1)} (x'_1)^{d_1}$, so does $\Res_{x_2}(p_1, p_2)$ by~\eqref{eq:poisson_expanded}.

    Finally, we will prove that $\Res_{x_2}(p_1, p_2)$ is in fact the minimal polynomial of the elimination ideal.
    Since the ideal $I_{\bg} \cap \mathbb{K}[x_1, x'_1, x''_1] = \ideal{f_{\min}}$ is principal, then
    by Lemma~\ref{lem03} 
    \[
    \Res_{x_2}(p_1, p_2) = c\cdot (f_{\min})^m, \text{ for some } c \in \mathbb{K}, m \in \mathbb{Z}_{> 0}.
    \]
    Assume $m \neq 1$. 
    Together with the decomposition~\eqref{eq:poisson}, this implies that $p_2(\alpha_i) = p_2(\alpha_j)$ for some $i \neq j$, so 
    \[
    \Rb \xi_i^{-1} a^{d_1 - 1} (1 + \xi_i^{d_2 } a^{d_2}) = \xi_j^{-1} a^{d_1 - 1} (1 + \xi_j^{d_2 } a^{d_2})
    \]
    and we obtain
    \[
    \Rb a^{d_2} = \dfrac{1 - \xi_{i} \xi_{j}^{-1}}{\xi_{i}(\xi_j^{d_2 - 1} - \xi_{i}^{d_2 - 1})}.
    \]
    Since $\xi_i$ and $\xi_j$ are distinct $d_1$-th roots of unity, 
    then $\xi_{i}(\xi_j^{d_2 - 1} - \xi_{i}^{d_2 - 1}) \neq 0$, 
    so $a^{d_2} \in \overline{\mathbb{K}}$. Since $a$ is transcendental over $\mathbb{K}$, we get a contradiction to $m \neq 1$. So $m = 1$ and $\Res_{x_2}(p_1, p_2) = f_{\min}$. 
 \end{proof}


\begin{lemma} \label{lem6} 
For the following polynomials $\bg = \tuple{g_1, g_2}$ in $\mathbb{K}[x_1, x_2]$ of degrees $d_1 > d_2 > 0$
    \begin{equation*}
            g_1(x_1, x_2)  = x_2^{d_1} + x_1x_2^{d_1-1},\quad \text{ and }\quad
            g_2(x_1, x_2) = x_2^{d_2},
    \end{equation*}
 the minimal polynomial $f_{\min}$  of $I_{\bg} \cap \KK[x_1^{(\infty)}]$ contains the monomial $(x'_1)^{2d_1-1}$. 
\end{lemma}

\begin{proof}
   Let $p_1 := x_1' - g_1(x_1, x_2)$ and
    \[
      p_2 := \mathcal{L}_{\bg}^\ast (p_1) =  x''_1 - d_1 x_2^{d_1+ d_2 -1} - x'_1 x_2^{d_1-1} - (d_1-1) x_1 x_2^{d_1+d_2-2}.
    \]
    
    We define $\Tilde{p}_i := p_i|_{x_1 = 0}$ for $i = 1, 2$.
    Since $p_1$ and $p_2$ are $x_2-$monic and $\deg_{x_2} \Tilde{p}_i = \deg_{x_2} p_i$ for $i = 1,2$, so
     \[
    \Res_{x_2}(\Tilde{p}_1,\Tilde{p}_2) = \Res_{x_2}(p_1, p_2)|_{x_1 = 0}.
    \]
    Since $\Tilde{p}_1$, as a polynomial in the variable $x_2$, has the roots $\alpha_i = \xi_i (x'_1)^{\frac{1}{d_1}}$, $1 \leqslant i \leqslant d_1$, where $\{1^{\frac{1}{d_1}} \} = \{\xi_1, \ldots, \xi_{d_1}\}$, then, by the Poisson formula~\cite[p.~427]{gelfand1994discriminants}, we have
    \[
        \Res_{x_2}(\Tilde{p}_1,\Tilde{p}_2) = \prod_{i=1}^{d_1} \Tilde{p}_2 (\alpha_i) = \prod_{i=1}^{d_1} \bigl(x''_1 - d_1\xi_i^{d_2 - 1} (x'_1)^{\frac{d_1 + d_2 - 1}{d_1}} - \xi_i^{- 1} (x'_1)^{\frac{2 d_1 - 1}{d_1}} \bigr).
    \]
    Since $d_1 > d_2$, then $\Res_{x_2}(\Tilde{p}_1,\Tilde{p}_2)$, as a polynomial in $\overline{\mathbb{K}(x''_1)}[(x'_1)^{\frac{1}{d_1}}]$, reaches the highest degree only in the term $(x'_1)^{2 d_1 - 1}$. Therefore, $\Res_{x_2}(\Tilde{p}_1,\Tilde{p}_2)$ contains the monomial $(x'_1)^{2 d_1 - 1}$, so $\Res_{x_2}(p_1, p_2)$ does.

    Now we will prove that $\Res_{x_2}(p_1, p_2)$ is in fact the minimal polynomial of the elimination ideal.
    Since the ideal $I_{\bg} \cap \mathbb{K}[x_1, x'_1, x''_1] = \ideal{f_{\min}}$ is principal, then
    by Lemma~\ref{lem03} 
    \[
    \Res_{x_2}(p_1, p_2) = c\cdot (f_{\min})^m, \text{ for some } c \in \mathbb{K}, m \in \mathbb{Z}_{> 0}.
    \]
    Assume $m \neq 1$, denote $(x'_1)^{\frac{1}{d_1}}$ by $a$ and replace the variable $x_1$ by $0$. 
    Then $\Tilde{p}_2(\alpha_i) = \Tilde{p}_2(\alpha_j)$ for some $i,j$
    \[
    x''_1 - d_1\xi_i^{d_2 - 1} a^{d_2} - \xi_i^{- 1} a^{d_1} = x''_1 -d_1\xi_j^{d_2 - 1} a^{d_2} - \xi_j^{- 1} a^{d_1}
    \]
    and we obtain
    \[
    a^{d_1 - d_2} = \dfrac{d_1(\xi_i^{d_2 - 1} - \xi_j^{d_2 - 1})}{\xi_j^{-1} - \xi_i^{-1}}.
    \]
    Since $\xi_i$ and $\xi_j$ are distinct $d_1$th roots of unity, then $\xi_j^{-1} - \xi_i^{-1} \neq 0$, so $a^{d_1 - d_2} \in \overline{\mathbb{K}}$.
    Since $a$ is transcendental, we get a contradiction to $m \neq 1$. So $m = 1$ and $\Res_{x_2}(p_2, p_1) = f_{\min}$.
\end{proof}

\begin{proof}[Proof of Theorem~\ref{thm:2d}]
The case $d_1 \leqslant d_2$ follows from Theorem~\ref{thm:sharp}.
Thus, in the rest of the proof we focus on the case $d_1 > d_2$.
In this case, the desired Newton polytope is a pyramid with the vertices corresponding to the monomials $1$, $x_1^{d_1(d_1 + d_2 - 1)}$, $(x_1')^{2d_1 - 1}$, $(x_1'')^{d_1}$, and~$x_1^{d_1(d_1 - 1)}( x_1')^{d_1}$ (see Figure~\ref{fig:2d}).

Consider the system $\bx' = \bg^\ast(\bx)$ from Lemma~\ref{lem5}.
The minimal polynomial $f_{\min}^\ast$ in this case contains monomials 
$x_1^{d_1(d_1 + d_2 - 1)}$, ${\Rb x_1^{d_1 (d_1 - 1)} (x_1')^{d_1}}$ and $(x_1'')^{d_1}$.
For $\varepsilon \in \KK$, we define $\bg^\ast_\varepsilon$ as the result of applying a transformation $x_1 \to x_1 + \varepsilon$ to $\bg^\ast$.
Since this transformation is invertible, it maps the minimal polynomial of $\bg^\ast$ to the minimal polynomial of $\bg^\ast_{\varepsilon}$.
That is, the latter is equal to $f_{\min, \varepsilon}^\ast := f_{\min}^\ast(x_1 + \varepsilon, x_1', x_1'')$.
We have proved that $f_{\min}^\ast$ contains a monomial $x_1^{d_1(d_1 + d_2 - 1)}$.
Therefore, the constant term of $f_{\min, \varepsilon}^\ast$ considered as a polynomial in $x_1, x_1', x_1''$ is a nonzero polynomial in $\varepsilon$.
Thus, there exists $\varepsilon^\ast \in \KK$ such that $f_{\min, \varepsilon^\ast}^\ast$ has a nonzero constant term.

We apply Proposition~\ref{prop:specific_to_generic} twice: to $\bg^\ast_{\varepsilon^\ast}$ constructed above and to $\bg^{**}$ from Lemma~\ref{lem6}.
We denote the resulting Zariski open sets $U_1, U_2 \supset V_{2, d_1} \times V_{2, d_2}$.
Consider an element $\bg \in U_1\cap U_2$.
Then the Newton polytope {\Rb of the minimal polynomial $f_{\min}$ for the system $\bg$} contains nonegative shifts $x_1^{d_1(d_1 + d_2 - 1)}$, $(x_1')^{2d_1 - 1}$, $(x_1'')^{d_1}$, and~$x_1^{d_1(d_1 - 1)}( x_1')^{d_1}$.
{\Rb Since these monomials correspond to the vertices} of the upper bound given 
by Theorem~\ref{theorem_general_specialized}, the only possible shift is the zero shift, 
{\Rb and thus} $f_{\min}$ for $\bg$ contains these monomials.
Since the shift of the Newton polygon of $f_{\min, \varepsilon}^\ast$ is zero, $f_{\min}$ also contains~$1$.
\end{proof}


\section{Algorithm}\label{sec:algorithm}
In this section we will use the bound from Theorem~\ref{theorem_general_specialized} to compute the minimal differential equation in $x_1$ for a system of differential equations of the form
\begin{equation} \label{eq::ODE}
    \bx' = \bg(\bx),
\end{equation}
where $\bx = \tuple{x_1, \ldots, x_n}$ and $\bg \in \KK[\bx]^n$. 

We begin with the first version of the algorithm which is randomized and, thus, may produce an incorrect result (for the probability analysis, see Proposition~\ref{prop:prob}).

\begin{algorithm}[H]
 \caption{(Randomized) computation of the minimal polynomial}
 \label{alg:fmin}
 \begin{algorithmic}[1]
    \Require  An ODE system
    \[
     \bx' = \bg(\bx),
    \]
    where $\bx = \tuple{x_1, \ldots, x_n}$ and $\bg \in \mathbb{Q}[\bx^{(\infty)}]$, an integer $R > 0$ (randomization parameter).
    \Ensure a polynomial $f \in \mathbb{Q}[x_1^{(\infty)}]$ for which one of the following holds
    \begin{itemize}
        \item $f$ is the minimal polynomial of $\ideal{ \bx' - \bg(\bx)}^{(\infty)} \cap \mathbb{Q}[x_1^{(\infty)}]$ (see~\eqref{id}) 
        \item or $f$ does not belong to $\ideal{\bx' - \bg(\bx)}^{(\infty)}$.
    \end{itemize}
    
    \State \label{line:order} $\nu \gets \operatorname{rank}(\frac{\partial}{\partial x_j}\mathcal{L}_{\bg}^{i - 1}(x_1))_{i, j = 1}^n$

    \State {\Ra Apply Theorem~\ref{theorem_general_specialized} to the ODE system and computed $\nu$ to obtain a set $S \subset \mathbb{Q}[x_1^{(\infty)}]$ of monomials such that $\supp f_{\min} \subset S $.
    Let $S := \{s_1, \ldots, s_{\ell} \}$.}
        
    \State $h_i\gets \mathcal{R}_{\bg}(s_i)$ (see Notation~\ref{not:proofs_notation}) for $i=1,\dots, \ell$.
    
    \State \label{line:points}Choose $\ell$ points $p_1, \ldots, p_\ell \in \mathbb{Z}^n$ by sampling uniformly at random from $[-R, R] \cap \mathbb{Z}$.
    
    \State \label{line:matrix}$M\gets (h_i(p_j))_{1\leqslant i,j \leqslant \ell}$ 
    \State $\mathcal{C} \leftarrow $ a basis of $\ker(M)$ 
    
    \State $F = \gcd \bigl(\{ \sum_{i=1}^{\ell}c_is_i \; | \; \mathbf{c} \in \mathcal{C} \} \bigr)$
   \Return $F$
 \end{algorithmic}
\end{algorithm}

\begin{proposition}
  \label{prop:prob}
    Algorithm~\ref{alg:fmin} terminates and {\Rb it is} correct.
    Furthermore, for any $\mathbf{g} \in \mathbb{Q}[\mathbf{x}]^n$ there is a proper Zariski closed subset $Z \subset \mathbb{Q}^{n\ell}$ such that, for all choices $\tuple{p_1, \ldots, p_{\ell}} \in \mathbb{Q}^{n \ell} \setminus Z$ in line \ref{line:points}, the output of Algorithm ~\ref{alg:fmin} is equal to the minimal polynomial of $\ideal{\bx' - \bg(\bx)}^{(\infty)} \cap \mathbb{Q}[x_1^{(\infty)}]$.
\end{proposition}

\begin{lemma}
\label{lem::minors}
    Let $ {\Ra q_1,\dots,q_s} \in \KK[\mathbf{x}]$ and denote by $V$ the vector space over $\KK$ spanned
    by the ${\Ra q_i}$'s. 
    Let $r:=\dim V\geqslant 1$ and introduce
    $s$ copies of $\mathbf{x}$, denoted $\mathbf{x}_1,\dots,\mathbf{x}_s$. 
    Let $N:=({\Ra q_i}(\mathbf{x}_j))_{1\leqslant i,j \leqslant s}$. Then there exists a nonzero $r \times r$-minor of $N$.
 
\end{lemma}

\begin{proof}
   We will show this via induction on $r$. W.l.o.g. assume that ${\Ra q_1,\dots,q_r}$ form a basis of $V$. 
   For the base case let us notice that ${\Ra q_1}\not \equiv 0$. 
   For the induction step we now show that the $r\times r$-minor $N_r$ of $N$ consisting of the first $r$ 
   rows and columns is nonzero.
   Laplace expansion around the first row yields
   \[
   N_r = \sum_{i=1}^r (-1)^{i + 1} {\Ra q_i}(\mathbf{x}_1) N_{i,r-1}
   \]
   with $N_{i,r-1}$ is the minor of $N$ with
   rows indexed by $\mathbf{x}_2,\dots,\mathbf{x}_r$ and columns indexed
   by ${\Ra q_1,\dots,q_{i-1},q_{i+1},\dots,q_r}$. By the induction hypothesis
   at least one of the $N_{i,r-1}$ is nonzero. Then, since all the $N_{i,r-1}$
   include only the variables in $\mathbf{x}_2,\dots,\mathbf{x}_r$, $N_r = 0$ implies 
   that there is a nontrivial $\KK$-linear relation between ${\Ra q_1,\dots,q_r}$, a contradiction.
\end{proof}

\begin{proof}[Proof of Proposition~\ref{prop:prob}]
  The termination of the algorithm is clear.

  To prove the correctness of the algorithm, we will use the notation
  of Algorithm \ref{alg:fmin} {\Ra and recall $I = \ideal{\bx' - \bg(\bx)}^{(\infty)}$} throughout the proof. 
  First we prove that the order of $f_{\min}$ is equal to $\nu$. 
  Lemma~\ref{lem::substitution} implies that the order of $f_{\min}$ equals to the transcendence degree of $\mathcal{R}_{\bg}(x_1^{(0)}),\dots,\mathcal{R}_{\bg}(x_1^{(n)})$ over $\mathbb{Q}$.
  By~\cite[Proposition 2.4]{ehrenborg1993apolarity} this transcendence degree is equal to the rank of the Jacobian of these polynomials.
  Since $\mathcal{R}_{\bg}(x_1^{(i)}) = \mathcal{L}_{\bg}^{i}(x_1)$ for every $i \geqslant 0$, this rank is equal to the one computed on line~\ref{line:order}.   
  
  Now we prove the correctness of the remaining part of the algorithm. Denote by
  $\mathbb{K}\langle S \rangle$ the $\mathbb{K}$-linear span of $S$. For any
  $f:= \sum\limits_{s\in S}\alpha_ss \in \mathbb{K}\langle S\rangle $ note that, Lemma
  \ref{lem::substitution} implies
  \[
  f\in I\; \Leftrightarrow \;\mathcal{R}_{\bg}(f) = 0 \; \Leftrightarrow \sum_{s\in S}\alpha_s\mathcal{R}_{\bg}(s) = 0\; \Leftrightarrow \; \sum_{i=1}^{\ell}\alpha_{s_i}h_i = 0. 
  \]
  Then, if we introduce $\ell$ duplicates of the
  variables $\bx$, denoted $\mathbf{x}_i$ for
  $i= 1,\ldots,\ell$, then the elements in
  $I\cap \mathbb{K}\langle S\rangle$ correspond to the kernel (in
  $\mathbb{K}^n$) of the matrix $N(\bx_1, \ldots, \bx_\ell):=(h_i(\mathbf{x}_j))_{1\leqslant i,j\leqslant \ell}$.
   Using this notation, matrix $M$ in Algorithm \ref{alg:fmin} can be written as
    $N(p_1, \ldots, p_\ell) =  \bigl( h_i(p_j) \bigr)_{1\leqslant i, j \leqslant \ell}$. 
  Hence $\operatorname{ker}(N)\subset \operatorname{ker}(M)$ with
  equality if and only if
  $r:=\operatorname{rk}(N) = \operatorname{rk}(M)$.

  Let $Z$ be the Zariski closed subset of
  $\mathbb{Q}^{n\ell}$ defined by the vanishing of the $r\times r$-minors of
  $N(\bx_1, \ldots, \bx_\ell)$. 
  By Lemma~\ref{lem::minors}, $Z$ is a proper Zariski closed subset of $\mathbb{Q}^{n\ell}$.  Then choosing  $\tuple{p_1,\dots,p_\ell}$ outside $Z$ will ensure the equality $\rk(N) = \rk(M)$.

  If this is the case, then the elements in the kernel of $M$ correspond
  to the elements in $I\cap \mathbb{K}\langle S\rangle$, {\Ra  and $\supp f_{\min} \subset S$ by
  Theorem~\ref{theorem_general_specialized}, thus}
  this kernel is nonempty and contains the minimal polynomial. 
  Then the gcd of a basis of this
  vector space gives the desired minimal polynomial.

  Finally, if we choose $\tuple{p_1, \ldots, p_{\ell}} \in Z$, then $\operatorname{ker}(N)\subsetneq \operatorname{ker}(M)$ and there exists an element in the kernel of $M$ that does not correspond to an element in the ideal $I \cap \mathbb{K}\langle S\rangle$. Then the gcd of a basis of $\operatorname{ker}(M)$ gives a polynomial that does not belong to the differential ideal $\ideal{\bx' - \bg(\bx)}^{(\infty)} \cap \mathbb{Q}[x_1^{(\infty)}]$.
\end{proof}

\begin{lemma} \label{lem::error}
    For fixed $\mathbf{g} \in \mathbb{Q}[\mathbf{x}]^n$ the Algorithm~\ref{alg:fmin} computes 
    {\Ra a ``wrong result'', i.e. a polynomial $f$ that does not belong to $\ideal{\bx' - \bg(\bx)}^{(\infty)}$} with probability at most $\mathcal{O}(\frac{1}{R})$ as $R \to \infty$.
\end{lemma}

\begin{proof}
    By Proposition~\ref{prop:prob} there exists a proper Zariski closed subset $Z$ of $\mathbb{Q}^{n\ell}$ s.t. only by choosing $\tuple{p_1, \ldots, p_{\ell}} \in Z$ with $p_i\in \mathbb{Q}^n$ in line \ref{line:points} the output of the Algorithm~\ref{alg:fmin} may be wrong.
    Choose any nonzero polynomial $P$ vanishing on $Z$, let $D := \deg P$.
    By the Demillo-Lipton-Schwartz-Zippel lemma \cite[Proposition 98]{zippel_effective_1993} the probability of $P$ being zero on a point in $\mathbb{Q}^{n\ell}$ with entries sampled uniformly at random from $[-R, R] \cap \mathbb{Z}$ for some integer $R$ does not exceed $\frac{D}{2R}$.
\end{proof}

{\Ra To make the individual steps of Algorithm~\ref{alg:fmin} more transparent, 
we illustrate its procedure using the toy example discussed in the Preliminaries (Example~\ref{ex::boytoy}).
 \begin{example}
    \begin{algorithmic}[1]
        \Require 
        \begin{equation*}
            \begin{cases}
                x_1' = x_2, \\
                x_2' = -x_1,
            \end{cases}
        \end{equation*}      
        and let $R = 29$.
        
        \State \label{line:order} $\nu = \operatorname{rank}
            \begin{pmatrix}
            \frac{\partial}{\partial x_1}(x_1) & \frac{\partial}{\partial x_2}(x_1)\\
            \frac{\partial}{\partial x_1}(x_2) & \frac{\partial}{\partial x_2}(x_2) 
            \end{pmatrix} = \operatorname{rank} 
            \begin{pmatrix}
                1 & 0\\
                0 & 1 
            \end{pmatrix} = 2.
            $
    
        \State $S = \{1, x_1, x_1', x_1'' \}$ using Theorem~\ref{theorem_general_specialized} 
        and $\nu = 2$.  In this case, the inequalities in Theorem~\ref{theorem_general_specialized} define 
        a three-dimensional simplex, and $S$ consists of four terms corresponding to the vertices of this simplex.
            
        \State $h_1 := \mathcal{R}_{\bg}(1) = 1, \quad h_2 := \mathcal{R}_{\bg}(x_1) = x_1, \quad
        h_3 := \mathcal{R}_{\bg}(x_1') = x_2, \quad h_4 := \mathcal{R}_{\bg}(x_1'') = -x_1$.
        
        \State \label{line:points} $p_1 = (2, 1), p_2 = (8, 8), p_3 = (0, 7), p_4 = (5, 1)$.
        
        \State \label{line:matrix}$M = 
        \begin{pmatrix}
          1 & 2 & 1 & -2 \\
          1 & 8 & 8 & -8 \\
          1 & 0 & 7 & 0 \\
          1 & 5 & 1 & -5 \\
        \end{pmatrix}$ 
        \State $\mathcal{C} = \ker(M) = (0, 1, 0, 1) $. 
        
        \State $\dim \ker(M) = 1$, thus, we do not need to compute the gcd of the elements corresponding 
        to the kernel, and  $F = x_1 + x_1''$.
        \Return $x_1 + x_1''$
     \end{algorithmic}
 \end{example}

}

Combining Algorithm~\ref{alg:fmin} with a membership check provided by Lemma~\ref{lem::substitution}, we can produce the following complete algorithm.

\begin{algorithm}[H]
 \caption{(Guaranteed) computation of the minimal polynomial}
 \label{alg:check}
 \begin{algorithmic}[1]
     \Require An ODE system
    \[
     \bx' = \bg(\bx),
    \]
    where $\bx = \tuple{x_1, \ldots, x_n}$ and $\bg \in \mathbb{Q}[\bx^{(\infty)}]$
    \Ensure A minimal polynomial $f_{\min}$ of $\ideal{\bx' - \bg(\bx)}^{(\infty)} \cap \mathbb{Q}[x_1^{(\infty)}]$ (see~\eqref{id})
    \State $R \gets 1893$
    \While{true}    
        \State Apply Algorithm~\ref{alg:fmin} to $\bx' = \bg(\bx)$ and $R$, denote the result by $F$
        \State $A \gets \mathcal{R}_{\bg}(F)$
        \If{$A = 0$}
        \Return $f_{\min} = F$
        \Else
        \State Set $R \gets 2R$
        \EndIf
    \EndWhile
    
\end{algorithmic}
\end{algorithm}

\begin{proposition}
    Algorithm~\ref{alg:check} is correct and terminates with probability one.
\end{proposition}

\begin{proof}
    The correctness of the algorithm follows from Proposition~\ref{prop:prob} and Lemma~\ref{lem::substitution}.
    By Lemma~\ref{lem::error}, there exists a positive real number $C$ such that the probability of the algorithm not terminating at the $i$-th iteration of the while loop is at most $\frac{C}{2^{i - 1}R}$, where $R = 1893$.
    Then the probability of the algorithm terminating is at least
    \[
    1 - \frac{C}{2 R} \cdot \frac{C}{4 R} \cdot \frac{C}{8 R} \cdot \ldots = 1.\qedhere
    \]
\end{proof}

\begin{remark} \label{remark::algorithm}
    In practice instead of computing the kernel of matrix $M$ in line \ref{line:matrix} of Algorithm~\ref{alg:fmin} over $\mathbb{Q}$ directly, we may use multi-modular arithmetic.
    We compute $M$ modulo several primes.
    Then the kernel is computed for each prime and the results are combined using  the Chinese Remainder Theorem and rational reconstruction to obtain
    the kernel over $\mathbb{Q}$, see for example \cite{rus_result}. 
    To be more efficient, we shrink the support bound after the first prime and then actually work with the exact support.
\end{remark}


\section{Implementation and performance}\label{sec:hard_examples}

We have produced a proof-of-concept implementation of the algorithm described in Section~\ref{sec:algorithm} in Julia language~\cite{bezanson2017julia}.
Our code relies on libraries Oscar \cite{OSCAR}, Nemo \cite{10.1145/3087604.3087611}, and Polymake \cite{gawrilow2000polymake}.
The source code of our implementation together with the instruction and the models used in this section is publicly available at
\[
\text{\url{https://github.com/ymukhina/Loveandsupport.git}}
\]

The goal of the present section is to show that our algorithm can perform differential elimination in reasonable time on a commodity hardware for some instances which are out of reach for the existing state-of-the-art software thus pushing the limits of what can be computed.
We demonstrate this using two sets of models:
\begin{itemize}
    \item \emph{Dense models}.
    For $n, d, D \in \mathbb{Z}_{> 0}$ and $a, b \in \mathbb{Z}$, by $\operatorname{Dense}_n(d, D, [a, b])$ we will denote a system of the form $\bx' = \bg(\bx)$, where the dimension of $\bx$ is $n$, $g_1$ is a random dense polynomial of degree $d$ and $g_2, \ldots, g_n$ are random dense polynomials of degree $D$, where the coefficients are sampled independently in random from ${a, a + 1, \ldots, b - 1, b}$.
    Here is, for example, an instance of $\operatorname{Dense}_3(3, 2, [1, 10])$:
    \begin{align*}
        x_1' &= 4 x_1^3 + 10 x_1^2 x_2 + 5 x_1^2 x_3 + 4 x_1^2 + x_1 x_2^2 + 5 x_1 x_2 x_3 + 2 x_1 x_2 + 3 x_1 x_3^2 + x_1 x_3 + 8 x_1 \\
        & \quad + 5 x_2^3 + 2 x_2^2 x_3 + 2 x_2^2 + 2 x_2 x_3^2 + 10 x_2 x_3 + x_2 + 5 x_3^3 + 4 x_3^2 + 5 x_3 + 8,\\
        x_2' &= 2 x_1^2 + 10 x_1 x_2 + 8 x_1 x_3 + 10 x_1 + x_2^2 + 6 x_2 x_3 + 4 x_2 + 2 x_3^2 + 4 x_3 + 10,\\
        x_3' &= 8 x_1^2 + 3 x_1 x_2 + 3 x_1 x_3 + 4 x_1 + 10 x_2^2 + 2 x_2 x_3 + 2 x_2 + 5 x_3^2 + 4 x_3 + 6.
    \end{align*}
    
    The specific randomly generated instances used for the experiments can be found in the repository.
    We note that the runtimes are not very sensitive to the choice of a particular random system.

    \item \emph{Sparse models}. 
    We use two specific ODE models.
    The first is a parametric model $\operatorname{BlueSky}$ exhibiting the blue-sky catastrophe phenomenon~\cite[Eq. (3)]{van2013triple}
    \begin{align*}
        & x_1' = \bigl(2 + a - 10(x_1^2 + x_2^2) \bigr)x_1 + x_2^2 + 2 x_2 + x_3^2, \\
        & x_2' = - x_3^3 - (1 + x_2)(x_2^2 + 2 x_2 + x_3^2) - 4x_1 + a x_2,\\
        & x_3' = (1 + x_2)x_3^2 + x_1^2 + b.
    \end{align*}
    We take the values of parameters $a = 0.456$ and $b = 0.0357$ as in~\cite{van2013triple}.

    The other model (we will call it $\operatorname{LV}$ for Lotka-Volterra) comes from the following parametric ODE system:
    \begin{align*}
       x_1' &= x_1(1 - a_{11}x_1 - a_{12}x_2 - a_{13}x_3) + b_1 x_2^2 + b_2 x_3^2,\\
       x_2' &= x_2(1 - a_{21}x_1 - a_{22}x_2 - a_{23}x_3 + c_1 x_2^3),\\
       x_3' &= x_3(1 - a_{31}x_1 - a_{32}x_2 - a_{33}x_3 + c_2 x_3^3).
    \end{align*}
    This model extends the standard three-species competition model corresponding to the case $b_1 = b_2 = c_1 = c_2 = 0$ with two generalized logistic growth~\cite[Eq. (6)]{LV} terms $c_1 x_2^2$ and $c_2 x_3^3$ (cf.~\cite[Eq. (1)]{SAMARDZIJA1988}) and recruitment-type terms $b_1 x_2^2$ and $b_2 x_3^2$ reminiscent to population models with stage structure (cf. \cite[Eq. (3)]{Pimm1987}).
    We do not claim any specific biological interpretation for this models but argue that it consists of primitives frequently used in population dynamics modeling.
    For the purposes of our computations experiments, we sampled the parameter values uniformly at random from $\left\{ \frac{1}{10}, \ldots, \frac{9}{10}, 1 \right\}$.
\end{itemize}

The software packages we use for comparison are:
\begin{itemize}    
    \item DifferentialThomas library in Maple~\cite{bachler2012algorithmic}.
    This library can compute differential Thomas decomposition for arbitrary polynomial PDE systems which, using an appropriate ranking, can be used to perform elimination.
    We used version Maple 2023.

    \item DifferentialAlgebra library (containing BLAD~\cite{blad}) which can compute a characteristic set decomposition for arbitrary polynomial PDE systems.
    Again, using an appropriate ranking, this can be used to perform elimination.
    We used version 1.11.
    
    \item StructuralIdentifiability~\cite{Dong2023} package written in Julia.
    It provides functionality for preforming differential elimination for ODE models in the state-space form~\eqref{eq:general_system} with rational dynamics.
    We used version 0.5.9. 
\end{itemize}

We report the performance of the selected software tools in Table \ref{tab:result}.

\begin{table}[htbp!]
\centering
	\begin{tabular}{ l| c c c c } 
	\hline
    Name & SI.jl & Maple(Diff.Thomas) & BLAD & Our (Algorithm~\ref{alg:check}) \\ 
	\hline
    $\operatorname{Dense}_3(2,3, [1,10])$ & 68 & $>$ 50h & OOM & 7\\

    $\operatorname{Dense}_3(2,3, [1,100])$ & 138 & $>$ 50h & OOM & 13 \\

    $\operatorname{Dense}_3(2,4, [1,10])$ & $>$ 50h & $>$ 50h & OOM & 414\\

    $\operatorname{Dense}_3(3,2, [1,10])$ & $>$ 50h & $>$ 50h & OOM & 215\\
    \hline

    $\operatorname{Dense}_4(1,2, [1,10])$ & 106 & $>$ 50h & OOM & 9\\

    $\operatorname{Dense}_4(1,2, [1,100])$ & 205 & $>$ 50h & OOM & 18\\

    $\operatorname{Dense}_4(2,1, [1,10])$ & OOM & $>$ 50h & OOM & 13\\

    $\operatorname{Dense}_4(2,1, [1,100])$ & OOM & $>$ 50h & OOM & 30\\

    \hline
	BlueSky & $>$ 50h & $>$ 50h & OOM & 317\\ 

    LV & 1537 & $>$ 50h & OOM & 114 \\

    \hline
    
\end{tabular}
\caption{Comparison with other approaches (times are in minutes if not written explicitly)\\ \emph{OOM = ``out of memory''}}\label{tab:result}
\end{table} 
   

\bibliographystyle{abbrvnat}
\bibliography{biblio}

\end{document}